%% file: article.tex
\documentclass[a4paper, 11pt]{article}
\usepackage{srcltx}
\usepackage{indentfirst}
\usepackage{graphicx}
\usepackage{overpic}
\usepackage{multirow}
\usepackage[dvipdfmx]{hyperref}
\usepackage[hang]{subfigure}
\usepackage{amsmath, amssymb, amsthm}
\usepackage{color}
\usepackage[mathscr]{euscript}

\newcommand\bbR{\mathbb{R}}
\newcommand\bbN{\mathbb{N}}
\newcommand\bxi{\boldsymbol{\xi}}
\newcommand\bx{\boldsymbol{x}}
\newcommand\bv{\boldsymbol{v}}
\newcommand\bu{\boldsymbol{u}}
\newcommand\bw{\boldsymbol{w}}

\newcommand\bF{\boldsymbol{F}}

\newcommand\dd{\,\mathrm{d}}
\newcommand\He{\mathit{He}}

\newcommand\NRxx{{NR{$xx$}}}
\newtheorem{theorem}{Theorem}

\newcommand\pd[2]{\dfrac{\partial {#1}}{\partial {#2}}}

\numberwithin{equation}{section}

\graphicspath{{images/}}

{\theoremstyle{remark} \newtheorem{remark}{Remark}}
\newtheorem{proposition}{Proposition}

\title{Solving Vlasov Equations Using {\NRxx} Method}

\author{Zhenning Cai\thanks{School of Mathematical Sciences, Peking
    University, Beijing, China, email: {\tt caizn@pku.edu.cn}.},~~ Ruo
  Li\thanks{HEDPS \& CAPT, LMAM \& School of Mathematical Sciences,
    Peking University, Beijing, China, email: {\tt
      rli@math.pku.edu.cn}.}~~ and Yanli Wang\thanks{CAPT, Beijing
    International Center for Mathematical Research, \& School of
    Mathematical Sciences, Peking University, Beijing, China, email:
    {\tt wangyanliwyl@gmail.com}.}}

\begin{document}
\maketitle
\input{article_abs_intro.tex}
\input{article_model.tex}
\input{article_vlasov.tex}
\input{article_num_ex.tex}
\input{article_conclusion.tex}
\bibliographystyle{plain}
\bibliography{../article}
\end{document}

%% file: article_abs_intro.tex
\begin{abstract}

  In this paper, we propose a moment method to numerically solve the
  Vlasov equations using the framework of the {\NRxx} method developed
  in \cite{NRxx,Cai,Li} for the Boltzmann equation. Due to the same
  convection term of the Boltzmann equation and the Vlasov equation,
  it is very convenient to use the moment expansion in the {\NRxx}
  method to approximate the distribution function in the Vlasov
  equations. The moment closure recently presented in \cite{Fan} is
  applied to achieve the globally hyperbolicity so that the local
  well-posedness of the moment system is attained. This makes our
  simulations using high order moment expansion accessible in the case
  of the distribution far away from the equilibrium which appears very
  often in the solution of the Vlasov equations. With the moment
  expansion of the distribution function, the acceleration in the
  velocity space results in an ordinary differential system of the
  macroscopic velocity, thus is easy to be handled. The numerical
  method we developed can keep both the mass and the momentum
  conserved. We carry out the simulations of both the Vlasov-Poisson
  equations and the Vlasov-Poisson-BGK equations to study the linear
  Landau damping. The numerical convergence is exhibited in terms of
  the moment number and the spatial grid size, respectively. The
  variation of discretized energy as well as the dependence of the
  recurrence time on moment order is investigated. The linear Landau
  damping is well captured for different wave numbers and collision
  frequencies. We find that the Landau damping rate linearly and
  monotonically converges in the spatial grid size. The results are in
  perfect agreement with the theoretic data in the collisionless case.

\vspace*{4mm}
\noindent {\bf Keywords:} Vlasov equations; {\NRxx} method; Landau
damping
\end{abstract}

\section{Introduction} \label{sec:intro} 

The Vlasov equation is a differential equation describing the time
evolution of the distribution function of plasma consisting of charged
particles with the long-range (for example, Coulomb) interaction. The
equation was first suggested for the description of plasma by A. Vlasov in
1938 \cite{Vlasov}. Due to the presence  of long-range Coulomb
interaction in the plasma, Vlasov started from the Vlasov equation,
which is a collisionless Boltzmann equation, and used a
self-consistent collective field created by the charged plasma particles
to get the Vlasov-Maxwell equations \cite{Henon}. The Vlasov-Poisson
(V-P) equations are an approximation of the Vlasov-Maxwell equations
in the nonrelativistic zero-magnetic field limit. The V-P equations
are used to describe various phenomena in plasma, in particular Landau
damping and the distributions in a double layer plasma. The
Vlasov-Pission-BGK (V-B) equations are also studied with a collision
term presented as the BGK term. They are the simplest kinetic
equations which correctly describe the essential features of collective
and dissipative (entropy-producing) particle interactions in
semiconductor plasmas \cite{BGK, BGK_1, Vojta, Mocker}.

Due to the complex phenomena in the plasma, numerical simulation plays
an important role in the study of the Vlasov and related
equations. There are several kinds of methods to solve the Vlasov
equations. The finite element method was proposed in
\cite{Finite-element1, Finite-element2}, which can be used to handle
complicated boundary problems but inconvenient to solve the high
dimension equations \cite{Filbet}. Meanwhile, the particle-in-cell
(PIC) method \cite{PIC} which used a finite number of macro-particles
to approximate the plasma is easy to be implemented, and the method to
discretize the Vlasov equations on a mesh of phase space was
introduced to remedy the problem in PIC that the inherent artificial
discrete particle noise made the description inaccuracy. The
semi-Lagrangian method \cite{SL} and the cubic interpolated
propagation method \cite{CIP} were also used to solve the Vlasov
equations. However, the first method destroyed the local characters
due to the reconstruction and the second one was quite expensive for
the storage of the distribution function and its derivatives. In
\cite{Filbet}, Filbet put forward a new method to deal with the force
term, which made the scheme  keep the mass and energy conserved.
Recently, an approach based on the moment method has been proposed in
\cite{UTH}, and therein the distribution function was expanded using
the Hermite polynomials with a prescribed macroscopic velocity chosen
as the expansion center and a prescribed thermal velocity as the
scaling factor at different locations.

In the past years, a regularized moment method was developed in
\cite{NRxx} to numerically solve the Boltzmann equation. This method
adopts the Hermite polynomial expansion to approximate the
distribution function, with the basis function shifted by the local
macroscopic velocity and scaled by the square root of the local
temperature. The approximated distribution function is used to
directly solve the Boltzmann equation without the deduction of the
moment system up to arbitrary order of moments. The method therein was
further explored as the {\NRxx} method \cite{Cai, Li} by introducing
the regularization term using asymptotic expansion in term of the mean
free path. Recently, a new regularization method \cite{Fan, Fan_new}
was derived with the guarantee that the regularized moment system is
globally hyperbolic. Due to the locally well-posedness provided by the
global hyperbolicity, it is eventually accessible that approximating
the distribution function far away from the equilibrium distribution
by the stable simulation using large number of moments. Inspired by
this progress, we in this paper develop the {\NRxx} method to study
the Vlasov equations, which is similar to the Boltzmann equation in
the convection term, while the distribution function is much farther
away from the equilibrium state than the gas flows, due to the long-range
Coulomb interactions.

Here we are focusing on the V-P and V-B equations. With the moment
expansion in the Vlasov equations, the convection part is smoothly
handled by the original {\NRxx} method for the Boltzmann equation. We
discretize the regularized term given in \cite{Fan_new} directly using
the central difference scheme. Our deduction shows that the electric
potential brings us an acceleration on the macroscopic velocity, thus
it is very convenient to be numerically integrated. Currently, the
collision term under our consideration is the simple BGK model to
avoid distraction. The approximated electric potential is obtained
by a three-point central scheme, and the numerical method keeps both
the mass and the momentum conserved. By the numerical resolution
study, it is exhibited that our method is numerically converged by the
comparison of Landau damping rates obtained using different spatial
grid size and number of moments. With the increasing of the number of
moments, the recurrence time is almost linearly related to the
square root of the moment expansion order. The discretized energy of
our method only varies slightly in comparison with the overall energy.
Since the Landau damping rate is affected by the wave number and the
collision frequency, different wave numbers and collision frequencies
are extensively studied. The wave numbers ranging from 0.2 to 0.5 are
simulated, and the collision frequencies are taken as 0.0 (the
collisionless case), 0.01 and 0.05. The results show that our method
can capture the linear Landau damping very well and the damping rates
obtained is in quantitative agreement with the theoretic data. We find
with surprise that the numerical Landau damping rate is linearly and
monotonically converged in term of the spatial grid size. Particularly,
the numerical damping rate converges to the theoretic data perfectly
in the collisionless case. This observation inspires us to predict the
Landau damping rates by the extrapolation of our numerical damping rates
with presence of the collision term.

The layout of this paper is as follows: in Section 2, the regularized
moment system is deduced for the Vlasov-BGK equations. In Section 3, we
present the detailed procedure of the numerical method. In Section 4,
the numerical examples including the numerical resolution study and
the linear Landau damping simulation with different parameters are
presented. Some concluding remarks are given in the last section.


%% file: article_model.tex

\section{Regularized Moment System}
Let $f(t,\bx,\bv)$, which depends on time $t$, position $\bx \in
\Omega \subset \bbR^3$ and the microscopic velocity $\bv \in \bbR^3$,
be the distribution function depicting the motion of particles. It is
governed by the V-B equations
\begin{equation}
  \label{eq:vlasov}
  \pd{f}{t} + \bv \cdot \nabla_{\bx} f + \bF(t,\bx,\bv) \cdot \nabla_{\bv} f
  = \nu(f_M-f),
\end{equation}
where $\nu(t,\bx)$ denotes the collision frequency,
and $\bF(t,\bx,\bv)$ is the electric force produced by the
self-consistent electric filed $\boldsymbol{E}(t,\bx)$:
\begin{equation}
\label{eq:force}
 \bF(t,\bx,\bv) = \frac{q}{m}\boldsymbol{E}(t,\bx), \quad \boldsymbol{E}(t,\bx) = -\nabla_{\bx}
  ~\phi(t,\bx),\quad  -\Delta_{\bx} \phi = q\frac{ \rho}{\epsilon_0},
\end{equation}
where $\phi(t,\bx)$ is the electric potential produced by the
particles; $\rho$, $q$, $m$ and $\epsilon_0$ stand for the density,
the single charge, the mass of one particle and the electric constant
respectively; $f_M$ is the Maxwellian defined as
\begin{equation}
  \label{eq:maxwellian}
  f_M =  \frac{\rho(t,\bx)}{(2\pi u_{th}(t,\bx))^{3/2}}
\exp\left(-\frac{|\bv-\bu(t,\bx)|^2}{2u_{th}(t,\bx)}\right),
\end{equation}
where the parameter $u_{th}(t,\bx)$ is the thermal velocity
\cite{UTH}, $\bu$ is the macroscopic velocity and $\rho(t,\bx)$ is
the same as that in \eqref{eq:force}. $f_M$ is related to $f$ by
\begin{equation}
  \label{eq:relationship}
\int_{\bbR^3} f(\bv)\left(
  \begin{array}{c}
    1 \\ \bv \\ |\bv|^2
  \end{array}\right)\mathrm{d}\bv = 
 \int_{\bbR^3} f_M(\bv)\left(
  \begin{array}{c}
    1 \\ \bv \\ |\bv|^2
  \end{array}\right)\mathrm{d}\bv.
\end{equation}
In the case of $\nu = 0$, we get the V-P equations. The relations
between the macroscopic variables and the distribution function are
deduced as
\begin{gather}
\label{eq:rho}
\rho(t,\bx) =  \int_{\bbR^3}f(t,\bx,\bv)\mathrm{d}\bv, \\
\label{eq:pri}
\rho(t,\bx)\bu(t,\bx) =
\int_{\bbR^3}\bv f(t,\bx,\bv)\mathrm{d}\bv, \\
  \label{eq:energy_1}
  \rho(t,\bx)|\bu(t,\bx)|^2 + 3\rho(t,\bx) u_{th}(t,\bx) 
=  \int_{\bbR^3}|\bv|^2 f(t,\bx,\bv)\mathrm{d}\bv. 
\end{gather}

The conservation of mass, momentum and total energy are all valid for
the V-B and V-P equations. Multiplying the equation \eqref{eq:vlasov} by
$1$ and $\bv$, direct integration with $\bv$ and $\bx$ gives us
\begin{gather}
  \label{eq:mass_conservtion}
  \frac{\mathrm{d}}{\mathrm{d}t}\int_{\bbR^3\times
    \bbR^3}f(t,\bx,\bv)\mathrm{d}\bx \mathrm{d}\bv = 0,
  \quad  t
  \in \bbR^+, \\
\label{eq:momentun_conservation}
 \frac{\mathrm{d}}{\mathrm{d}t}\int_{\bbR^3\times
    \bbR^3}\bv f(t,\bx,\bv)\mathrm{d}\bx \mathrm{d}\bv = 0,
  \quad  t
  \in \bbR^+.
\end{gather} 
Multiplying the equation \eqref{eq:vlasov} by $|\bv|^2$ and
integrating by parts, we get the conservation of the total energy for
the system \eqref{eq:vlasov} and \eqref{eq:force}:
\begin{equation}
  \label{eq:energy_conversvation}
  \frac{\mathrm{d}}{\mathrm{d}t}\left(\int_{\bbR^3\times
    \bbR^3}f(t,\bx,\bv)|\bv|^2\mathrm{d}\bx\mathrm{d}\bv +
  \int_{\bbR^3}|\boldsymbol{E}(t,\bx)|^2\mathrm{d}\bx \right) = 0, \quad  t
  \in \bbR^+.
\end{equation}

\subsection{Hermite expansion of the distribution function} 
\label{sec:discretization}
Following the method in \cite{NRxx, NRxx_new}, we expand the
distribution function into Hermite series as
\begin{equation}
  \label{eq:expansion}
  f(\bv) = \sum_{\alpha \in \bbN^3}f_{\alpha}\mathcal{H}_{u_{th},\alpha}(\bxi),
\end{equation}
where $\alpha=(\alpha_1,\alpha_2,\alpha_3)$ is a three-dimensional
multi-index, and 
\begin{equation}
  \label{eq:scaling}
  \bxi = \frac{\bv-\bu}{\sqrt{u_{th}}}.
\end{equation}
The basis functions $\mathcal{H}_{u_{th},\alpha}$ are defined
as 
\begin{equation}
  \label{eq:base}
  \mathcal{H}_{u_{th},\alpha}(\bxi) =\prod\limits_{d=1}^3
 \frac{1}{\sqrt{2\pi}}u_{th}^{-\frac{\alpha_d+1}{2}}
\He_{\alpha_d}(\xi_d)\exp \left(-\frac{\xi_d^2}{2} \right),
\end{equation}
where $\He_{\alpha_d}$ is the Hermite polynomial
\begin{equation}
  \label{eq:hermit}
  \He_n(x) = (-1)^n\exp \left( \frac{x^2}{2} \right) \frac{\dd^n}{\dd x^n} \exp \left(-\frac{x^2}{2} \right).
\end{equation}
For convenience, $\He_{n}(x)$ is taken as zero if $n<0$, thus
$\mathcal{H}_{u_{th},\alpha}(\bxi)$ is zero when any 
component of $\alpha$ is negative.  With such an expansion, the
Maxwellian $f_{M}$  can be written as  
 \begin{equation}
    f_M(\bv) = \rho\mathcal{H}_{u_{th},0}(\bxi),
  \end{equation}  
and the BGK collision term is written as
\begin{equation}
  \label{eq:collision}
  \nu(f_M-f) = -\nu\sum_{|\alpha| \geq 1}f_{\alpha}\mathcal{H}_{u_{th},\alpha}(\bxi).
\end{equation}

The definition of $\bxi$ shows that each basis function is an
exponentially decreasing function multiplied by a multi-dimensional
Hermite polynomial shifted by the local macroscopic velocity $\bu$ and
scaled by the square root of the local thermal velocity $u_{th}$. For
any vector $\bu'$ and positive number $u_{th}'$, if the distribution
function $f$ is expanded as
\begin{equation}
\label{eq:expansion_2}
f(\bv) = \sum\limits_{\alpha \in \bbN^3 }f_{\alpha}'\mathcal{H}_{u_{th}',\alpha}(\bxi'), \quad
\bxi' = \frac{\bv-\bu'}{\sqrt{u_{th}'}},
\end{equation}
then the following relations hold
\begin{subequations}
\label{eq:similar_relation} 
 \begin{align}
    \rho &=  f_0', \\
\rho \bu &=  \rho \bu' + (f_{e_d}')_{d=1,2,3}^T, \\
\rho|\bu|^2 + 3\rho u_{th} &=  2\rho \bu \cdot \bu' - \rho|\bu'|^2 +
\sum\limits_{d=1}^3(u_{th}'f_0' + 2f_{2e_d}').
\end{align}
\end{subequations}
In the case of $\bu' = \bu$ and $u_{th}' = u_{th}$ in
\eqref{eq:similar_relation}, the following relations between the
coefficients $f_{\alpha}$ can be verified:
\begin{equation}
\label{eq:estimate}
  f_0 = \rho(t,\bx), \quad f_{e_i} = 0, \quad
  \sum\limits_{d=1}^3 f_{2e_d} = 0, \quad i =1, 2 , 3.
\end{equation} 
Moreover, direct calculations give us the relations between the
coefficients  $f_{\alpha}$ in \eqref{eq:expansion} as
\begin{gather}
  \label{eq:q_p}
q_{i} = 2f_{3e_i} + \sum\limits_{d=1}^3f_{2e_d+e_i}, \\
p_{ij}  -\frac{1}{3}\delta_{ij}\sum\limits_{d=1}^3p_{dd} =
(1+\delta_{ij})f_{e_i+e_j}.
\end{gather} 
where $i,j = 1,2,3$, $q_i$ and $p_{ij}$ are related to $f$
by
\begin{gather}
  q_i  = \frac{1}{2}\int_{\bbR^3}|\bv - \bu|^2(v_i - u_i)f \mathrm{d}\bv,\\
p_{i,j} = \int_{\bbR^3}(v_i - u_i)(v_j - u_j)f
\mathrm{d}\bv.
\end{gather}

\subsection{Moment expansion of Vlasov equation}
\label{moment_equations}
To get the moment system, we substitute the expansion
\eqref{eq:expansion} into \eqref{eq:vlasov} and then match the
coefficients for the same basis functions. Taking the temporal and
spatial derivatives directly on the basis functions
$\mathcal{H}_{u_{th},\alpha}$, the term with the expansion
\eqref{eq:expansion}
\[
\pd{f}{t} + \bv \cdot \nabla_{\bx} f
\]
is expanded as
\begin{equation}
\begin{split}
&~~\sum_{\alpha \in \bbN^3} \Bigg\{ \left(
    \frac{\partial f_{\alpha}}{\partial t} +
    \sum_{d=1}^3 \frac{\partial u_d}{\partial t} f_{\alpha-e_d} +
    \frac{1}{2} \frac{\partial u_{th}}{\partial t}
      \sum_{d=1}^3 f_{\alpha-2e_d}
  \right) \\
& \qquad + \sum_{j=1}^3 \Bigg[ \left(
    u_{th} \frac{\partial f_{\alpha - e_j}}{\partial x_j} +
    u_j \frac{\partial f_{\alpha}}{\partial x_j} +
    (\alpha_j + 1) \frac{\partial f_{\alpha+e_j}}{\partial x_j}
  \right) \\
& \qquad \qquad + \sum_{d=1}^3 \frac{\partial u_d}{\partial x_j}
    \left( u_{th} f_{\alpha-e_d-e_j} + u_j f_{\alpha-e_d}
      + (\alpha_j + 1) f_{\alpha-e_d+e_j} \right) \\
& \qquad \qquad + \frac{1}{2} \frac{\partial u_{th}}{\partial x_j}
    \sum_{d=1}^3 \left(
      u_{th} f_{\alpha-2e_d-e_j} + u_j f_{\alpha-2e_d}
      + (\alpha_j + 1) f_{\alpha-2e_d+e_j}
    \right)
  \Bigg] \Bigg\} \mathcal{H}_{u_{th},\alpha} \left(
    \frac{\bxi - \bu}{\sqrt{u_{th}}}
  \right),
\end{split}
\end{equation}
the acceleration term $\bF \cdot \nabla_{\bv}f$ is expanded as
\begin{equation}
-\sum_{\alpha \in \bbN^3} \sum_{d=1}^3 F_d f_{\alpha-e_d}
  \mathcal{H}_{u_{th},\alpha} \left(
    \frac{\bxi - \bu}{\sqrt{u_{th}}}
  \right),
\end{equation}
and the collision term $\nu(f_M - f)$ is expanded as
\begin{equation}
\nu\sum_{|\alpha|>1}f_{\alpha}\mathcal{H}_{u_{th},\alpha} \left(
    \frac{\bxi - \bu}{\sqrt{u_{th}}}
  \right).  
\end{equation}
Substituting these expansions back into \eqref{eq:vlasov}, and collecting
coefficients for the same basis functions, we get the following
general moment equations with a slight rearrangement:
\begin{equation} \label{eq:mnt_eq}
\begin{split}
& \frac{\partial f_{\alpha}}{\partial t} +
  \sum_{d=1}^3 \left(
    \frac{\partial u_d}{\partial t} +
    \sum_{j=1}^3 u_j \frac{\partial u_d}{\partial x_j} - F_d
  \right) f_{\alpha-e_d} + \frac{1}{2} \left(
    \frac{\partial u_{th}}{\partial t} +
    \sum_{j=1}^3 u_j \frac{\partial u_{th}}{\partial x_j}
  \right) \sum_{d=1}^3 f_{\alpha-2e_d} \\
& \quad + \sum_{j,d=1}^3 \left[
    \frac{\partial u_d}{\partial x_j} \left(
      u_{th} f_{\alpha-e_d-e_j} + (\alpha_j + 1) f_{\alpha-e_d+e_j}
    \right) + \frac{1}{2} \frac{\partial u_{th}}{\partial x_j} \left(
      u_{th} f_{\alpha-2e_d-e_j} + (\alpha_j + 1) f_{\alpha-2e_d+e_j}
    \right)
  \right] \\
& \quad + \sum_{j=1}^3 \left(
    u_{th} \frac{\partial f_{\alpha - e_j}}{\partial x_j} +
    u_j \frac{\partial f_{\alpha}}{\partial x_j} +
    (\alpha_j + 1) \frac{\partial f_{\alpha+e_j}}{\partial x_j}
  \right) = \nu(1 -
  \delta(\alpha))f_{\alpha},
\end{split}
\end{equation}
where $\delta(\alpha)$ is defined by 
\begin{equation}
\label{eq:delta}
  \delta(\alpha) = \left\{
    \begin{array}{ll}
      0, & \text{if}~~ |\alpha| \geqslant 2, \\
      1, & \text{otherwise.}
    \end{array}
\right.
\end{equation} 
Following the method in \cite{Li}, we deduce the mass conservation in
the case of $\alpha=0$ as
 \begin{equation}
\label{mass_con}
  \frac{\partial f_0}{\partial x_j} +
  \sum_{j=1}^3 \left(
    u_j \frac{\partial f_0}{\partial x_j} +
    f_0 \frac{\partial u_j}{\partial x_j}
  \right) = 0.
 \end{equation}
If we set $\alpha = e_d$, with $ d = 1,2,3$,  \eqref{eq:mnt_eq} reduces to
\begin{equation}
\label{eq:alpha = e_d}
f_0 \left(
  \frac{\partial u_d}{\partial t} +
  \sum_{j=1}^3 u_j \frac{\partial u_d}{\partial x_j} - F_d
\right) + f_0 \frac{\partial u_{th}}{\partial x_d} 
  + u_{th} \frac{\partial f_0}{\partial x_d}
  + \sum_{j=1}^3 (\delta_{jd} + 1)
    \frac{\partial f_{e_d + e_j}}{\partial x_j} = 0,
\end{equation}
which is simplified as
\begin{equation} \label{eq:mtm}
f_0 \left(
  \frac{\partial u_d}{\partial t} +
  \sum_{j=1}^3 u_j \frac{\partial u_d}{\partial x_j} - F_d
\right) + \sum_{j=1}^3 \frac{\partial p_{jd}}{\partial x_j} = 0.
\end{equation}

Then we consider the case of $|\alpha| \ge 2$. Multiplying $|\bv -
\bu|^2$ on both sides of \eqref{eq:vlasov}, and integrating with
respect to $\bv$ on $\bbR ^3$, we have
\begin{equation} 
\label{eq:moment_eq}
f_0 \left( \frac{\partial u_{th}}{\partial t}
  + \sum_{j=1}^3 u_j \frac{\partial u_{th}}{\partial x_j} \right)
  + \frac{2}{3} \sum_{j=1}^3 \left(
    \frac{\partial q_j}{\partial x_j} +
    \sum_{d=1}^3 p_{jd} \frac{\partial u_d}{\partial x_j}
  \right) = 0.
\end{equation}
\begin{remark}
 Since  
\begin{equation}
\int_{\bbR^3} |\bv-\bu|^2\pd{f}{v_i}\mathrm{d}\bv =
-2\int_{\bbR^3}(v_i - u_i)f \mathrm{d}\bv  = 0, \quad i = 1,2,3,
\end{equation} 
the acceleration term does not appear in \eqref{eq:moment_eq}.
\end{remark}
Substituting \eqref{eq:mtm} and \eqref{eq:moment_eq} into
\eqref{eq:mnt_eq}, we eliminate the temporal derivatives of $\bv$
and $u_{th}$. Then the quasi-linear form of the moment system
reads:
\begin{equation} \label{eq:mnt_system}
\begin{split}
& \frac{\partial f_{\alpha}}{\partial t} - \frac{1}{f_0}
  \sum_{d=1}^3 \sum_{j=1}^3
    \frac{\partial p_{jd}}{\partial x_j} f_{\alpha-e_d}
  - \frac{1}{3f_0} \sum_{j=1}^3 \left(
    \frac{\partial q_j}{\partial x_j} +
    \sum_{d=1}^3 p_{jd} \frac{\partial u_d}{\partial x_j}
  \right) \sum_{d=1}^3 f_{\alpha-2e_d} \\
& \quad + \sum_{j,d=1}^3 \left[
    \frac{\partial u_d}{\partial x_j} \left(
      u_{th} f_{\alpha-e_d-e_j} + (\alpha_j + 1) f_{\alpha-e_d+e_j}
    \right) + \frac{1}{2} \frac{\partial u_{th}}{\partial x_j} \left(
      u_{th} f_{\alpha-2e_d-e_j} + (\alpha_j + 1) f_{\alpha-2e_d+e_j}
    \right)
  \right] \\
& \quad + \sum_{j=1}^3 \left(
    u_{th} \frac{\partial f_{\alpha - e_j}}{\partial x_j} +
    u_j \frac{\partial f_{\alpha}}{\partial x_j}\right) +
    \sum_{j=1}^3 (\alpha_j + 1) \frac{\partial f_{\alpha+e_j}}{\partial x_j}
   =\nu(1 - \delta(\alpha))f_{\alpha}, \quad \forall |\alpha| \geq 2.
\end{split}
\end{equation}
We collect the equations \eqref{mass_con}, \eqref{eq:mtm}, \eqref{eq:moment_eq} and
\eqref{eq:mnt_system} together to obtain a moment system with infinite
number of equations.

\subsection{Closure of the moment system}
\label{the moment closure}
With a truncation of \eqref{eq:expansion}, \eqref{eq:mnt_system} will
result in a finite moment system. Precisely, we let $M\geqslant 3$ be
a positive integer and only the coefficients in the set $\mathcal{M}
= \{f_{\alpha}\}_{|\alpha|\leqslant M}$ are considered. Let
$F_{M}(\bu, u_{th})$ denotes the linear space spanned by all
$\mathcal{H}_{u_{th}, \alpha}(\bxi)$'s with $|\alpha|\leqslant M$,
and the expansion \eqref{eq:expansion} is truncated as
\begin{equation}
  \label{eq:expansion_1}
  f(\bv) \approx \sum_{|\alpha| \leqslant M} f_{\alpha} \mathcal{H}_{u_{th}, 
    \alpha}\left(\frac{\bv -\bu}{\sqrt{u_{th}}}\right),
\end{equation}
with $f(\bv) \in F_{M}(\bu, u_{th}) $ and $f_{\alpha} \in
\mathcal{M}$. The moment equations which contain
$\partial{f_{\alpha}}/\partial{t}$ with $|\alpha|>M$ are disregarded
in \eqref{eq:mnt_system}.  Then, \eqref{mass_con},
\eqref{eq:mtm} and \eqref{eq:mnt_system} with $2\leqslant
|\alpha|\leqslant M$ lead to a system with finite number of equations.

Due to the presence of the terms with ${\partial f_{\alpha +
    e_j}}/{\partial x_j}$, $|\alpha| = M$, the moment system we
obtained is not closed yet. We rewrite \eqref{eq:mnt_system} into the
form below:
\begin{equation}
\label{eq:mnt_system_2}
\pd{f_{\alpha}}{t} + \mathcal{A}_{\alpha} +  \mathcal{B}_{\alpha}
= \nu(1 - \delta(\alpha))f_{\alpha},
\end{equation}
where in the case of $2 \leq |\alpha| < M$,
\begin{equation*}
\begin{split}
\mathcal{A}_{\alpha}& = - \frac{1}{f_0}
  \sum_{d=1}^3 \sum_{j=1}^3
    \pd{p_{jd}}{x_j} f_{\alpha-e_d} + \cdots
  + \sum_{j=1}^3 \left(
    u_{th} \pd{f_{\alpha - e_j}}{x_j} +
    u_j \pd{f_{\alpha}}{x_j}\right) \\
 & +\sum_{j=1}^3 (\alpha_j + 1) \pd{ f_{\alpha+e_j}}{x_j}, \\
 \mathcal{B}_{\alpha}& = 0,
\end{split}
\end{equation*}
and in the case of $|\alpha| = M$,
\begin{equation*}
\begin{split}
  \mathcal{A}_{\alpha}& = - \frac{1}{f_0} \sum_{d=1}^3 \sum_{j=1}^3
  \pd{p_{jd}}{x_j} f_{\alpha-e_d} + \cdots
   + \sum_{j=1}^3 \left( u_{th} \pd{f_{\alpha - e_j}}{x_j} +
    u_j \pd{f_{\alpha}}{x_j}\right), \\
  \mathcal{B}_{\alpha}& = \sum_{j=1}^3\mathcal{B}_{\alpha,j}, \quad
  \mathcal{B}_{\alpha,j} = (\alpha_j + 1) \pd{ f_{\alpha+e_j}}{x_j}.
\end{split}
\end{equation*}

Clearly, the moments $f_{\alpha+e_j}$ in term $\mathcal{B}_{\alpha}$
with $|\alpha| = M$ are not in the set of moments $\mathcal{M}
=\{f_{\alpha}\}_{\alpha \leqslant M}$, and have to be substituted by
some expressions consisting of lower order moments to make the moment
system closed. If we simply let $\mathcal{B}_{\alpha} = 0$, the
Grad-type system associated with the moment set $\mathcal{M} =
\{f_{\alpha}\}_{|\alpha| \leqslant M}$ is obtained. It is known that
the Grad-type system is not locally well-posed due to the lack of the
global hyperbolicity, which results in numerical blow-up when the
distribution function is far away from the equilibrium state. The
regularization method in \cite{Fan_new, Fan} is adopted here to
achieve a globally hyperbolic moment closure. The $(M+1)$-st order
terms are substituted as bellow
\begin{equation}
  \label{eq:revision}
  \frac{\partial f_{\alpha+e_j}}{\partial x_j} \longrightarrow
  -\left(\sum\limits_{d=1}^3f_{\alpha-e_d+e_j}\pd{u_d}{x_j} +
  \frac{1}{2}\left(\sum_{d=1}^Df_{\alpha-2e_d+e_j}\pd{u_{th}}{x_j}\right)\right),
  \quad |\alpha| = M,~~j=1,2,3.
\end{equation}
Let $\mathcal{\hat{B}}_{\alpha}$ denote the regularization term based
on the characteristic speed correction in \cite{Fan_new, Fan} for
$|\alpha| = M$
\begin{equation}\label{RM}
  \mathcal{\hat{B}}_{\alpha} = -\sum_{d=1}^3\mathcal{\hat{B}}_{\alpha,j}, 
  \quad \mathcal{\hat{B}}_{\alpha,j} = -\left(\sum\limits_{d=1}^3f_{\alpha-e_d+e_j}
  \pd{u_d}{x_j} + \frac{1}{2}\sum_{d=1}^Df_{\alpha-2e_d+e_j}
    \pd{u_{th}}{x_j}\right).
\end{equation}
Substituting the $(M+1)$-st order term $\mathcal{B}_{\alpha}$ with the
regularization term $\mathcal{\hat{B}}_{\alpha}$, the moment equations
are revised as
\begin{equation}
\label{eq:mnt_system_4}
\pd{f_{\alpha}}{t} + \mathcal{A}_{\alpha} + \mathcal{\hat{B}}_{\alpha}  = \nu(1 -
  \delta(\alpha))f_{\alpha},
\end{equation}
with $\mathcal{\hat{B}}_{\alpha} = \mathcal{B}_{\alpha} = 0$ for $2
\leqslant |\alpha| < M$.  If the distribution function $f$ only
depends on $x_1$ in the spatial direction, we have that
\begin{equation}
  \mathcal{\hat{B}}_{\alpha,j}=0, \quad {\rm for~~} j = 2, 3.
\end{equation}
\begin{remark}
  The regularized moment system \eqref{eq:mnt_system_4} is not able to
  be written into a conservation law, for the presence of the
  regularization term. If we let $\mathcal{\hat{B}}_{\alpha}=0$ with
  $|\alpha| = M$, the system changes into the conservative Grad-type
  system.
\end{remark}


%% file: article_vlasov.tex
\section{Numerical Method}
The numerical scheme for the regularized moment system
\eqref{eq:mnt_system_4} is a natural extension of the method in
\cite{NRxx}. By a standard fraction step method, we split the
V-B equations into the following parts:
\begin{itemize}
\item the convection step: 
\begin{equation}
\label{eq:convection_part}
\pd{f}{t} + \bv \cdot \nabla_{\bx}f =  0,
\end{equation}
\item the acceleration step: 
\begin{gather}
\label{eq:force_part_1}
\pd{f}{t} + \bF(t,\bx,\bv)\cdot \nabla_{\bv} f = 0,\\
\bF(t,\bx,\bv) = \frac{q}{m}\boldsymbol{E}(t,\bx), \quad
\boldsymbol{E}(t,\bx) = -\nabla_{\bx} ~\phi(t,\bx),\quad
-\Delta_{\bx} \phi = q\frac{ \rho}{\epsilon_0}.
\end{gather}
\item the collision step:
\begin{equation}
\label{eq:collision_part}
\pd{f}{t} = \nu (f_M - f).
\end{equation}
\end{itemize}
We observe that only \eqref{eq:mtm} contains the electric force $\bF$
in the governing equations. Thus the governing equations of the
acceleration part turn into
\begin{gather}
  \label{eq:force_part} 
\partial_t\bu = \bF,\\
  \bF(t,\bx,\bv) = \frac{q}{m}\boldsymbol{E}(t,\bx), \quad \boldsymbol{E}(t,\bx) = -\nabla_{\bx}
  ~\phi(t,\bx),\quad  -\Delta_{\bx} \phi = q\frac{ \rho}{\epsilon_0}.
\end{gather}

Here we restrict our study in 1D spatial space. The numerical scheme
adopted in the $x$-direction is the standard finite volume
discretization. Suppose $\Gamma_h$ to be a uniform mesh in $\bbR$, and
each cell is identified by an index $j$. For a fixed $x_0 \in \bbR$
and $\Delta x > 0$,
\begin{equation}
  \Gamma_h = \big\{T_{j} = x_0 + \left( j\Delta x, ~(j + 1)\Delta x
  \right): j \in  \mathbb{Z}\big\}.
\end{equation} 
The numerical solution which is  the approximation of the distribution
function $f$ at  $t = t_n$ is denoted as
\begin{equation}
  f_h^n(x, \bv) = f_j^n(\bv),\quad x \in T_j,
\end{equation}
where $f_{j}^n(\bv) $ is the approximation over the cell $T_{j}$ at
the $n$-th time step and has the following Hermite expansion form as
\[
f_{j}^n(\bv) = \sum_{|\alpha| \leqslant M} f_{\alpha,j}^n
\mathcal{H}_{u_{th,j}^n, \alpha}\left(\frac{\bv
    -\bu_j^n}{\sqrt{u_{th,j}^n}}\right).
\]
 
\subsection{The convection step }
The equation \eqref{eq:convection_part} is discretized as
\begin{equation}
\label{eq:scheme}
  f_{j}^{n+1,\ast}(\bv) = f_{j}^n(\bv) + K_{1,j}^n(\bv) + K_{2,j}^{n}(\bv),
\end{equation}
where $K_{1,j}^n$ is the contribution of the term $\mathcal{A}_\alpha$
in \eqref{eq:mnt_system_4} without considering the acceleration, and
$K_{2,j}^n$ is the contribution of the term $\mathcal{\hat{B}}_\alpha$
in \eqref{eq:mnt_system_4}. Noticing that the term $\mathcal{A}_\alpha$
results in the conservative part in the Grad-type moment system,
its contribution $K_{1,j}^n$ is discretized in the conservative
formation as
\begin{equation}
\label{eq:con}
K_{1,j}^n(\bv) = -\frac{\Delta t^n}{\Delta x}
\left[F_{j+\frac{1}{2}}^n(\bv) - F_{j-\frac{1}{2}}^n(\bv)\right],
\end{equation}
where $F_{j+\frac{1}{2}}^n$ is the numerical flux between cell $T_{j}$
and $T_{j+1}$ at $t^n$. We use the HLL scheme in our numerical
experiments following \cite{Cai}, which reads:
\begin{equation}
  \label{eq:HLL}
  F_{j+\frac{1}{2}}^n(\bv) =
  \begin{cases}
    v_1f_{j}^n(\bv), & 0 \leqslant \lambda_{j+\frac{1}{2}}^L, \\
   \dfrac{\lambda_{j+\frac{1}{2}}^Rv_1f_{j}^n(\bv) -
     \lambda_{j+\frac{1}{2}}^L v_1f_{j+1}^n(\bv)+
     \lambda_{j+\frac{1}{2}}^L\lambda_{j+\frac{1}{2}}^R[f_{j+1}^n(\bv) -
     f_{j}^n(\bv)]}{\lambda_{j+\frac{1}{2}}^R-\lambda_{j+\frac{1}{2}}^L}, &
   \lambda_{j+\frac{1}{2}}^L < 0 < \lambda_{j+\frac{1}{2}}^R, \\
 v_1f_{j+1}^n(\bv),  & 0 \geqslant \lambda_{j+\frac{1}{2}}^R,
  \end{cases}
\end{equation}
where $\lambda_{j+\frac{1}{2}}^L$ and $\lambda_{j+\frac{1}{2}}^R$ are
the fastest signal speeds \cite{Fan} as
\begin{equation}
  \label{eq:velocity}
\begin{split}
  \lambda_{j+\frac{1}{2}}^L = 
  \min\{u_{1,j}^n-C_{M+1} \sqrt{u_{th,j}^n},
  u_{1, j+1}^n - C_{M+1}\sqrt{u_{th,j+1}^n} \}, \\
  \lambda_{j+\frac{1}{2}}^R = \max\{u_{1,j}^n + C_{M+1} \sqrt{u_{th,j}^n},
  u_{1,j+1}^n + C_{M+1}\sqrt{u_{th,j+1}^n} \},
\end{split}
\end{equation}
where $C_{M+1}$ is the greatest zero of $He_{M+1}(x)$, $u_1$ is the
first component of the macroscopic velocity $\bu$, and $u_{th}$ is the
thermal velocity. The formula of the signal speed is also used to
determine the time step $\Delta t^n$ by the CFL condition.
Two points remain unclear in the numerical flux. The first one is how
to calculate $v_1f_{j}^n(\bv)$. This is managed according to the
recursion relation of Hermite polynomials:
\begin{equation}
  \label{eq:v_multiply}
\begin{split}
 v_1f_{j}^n(\bv)& =
  [(u_{th,j}^n)^{1/2}(\xi_{1,j}^n) +
  (u_{1,j}^n)]\sum\limits_{|\alpha|\leqslant
    M}f_{j,\alpha}^n\mathcal{H}_{j,\alpha}^n(\bxi_{j}^n)
  \\
&= \sum\limits_{|\alpha|\le M}
f_{j,\alpha}^n[u_{th,j}^n\mathcal{H}_{j,\alpha+e_1}^n(\bxi_{j}^n)
+ (u_{1,j}^n)\mathcal{H}_{j,\alpha}^n(\bxi_{j}^n) +
\alpha_1\mathcal{H}_{j,\alpha-e_1}^n(\bxi_{j}^n)],
\end{split}
\end{equation} 
where $\bxi_{j}^n = (\bv -\bu_{j}^n)/\sqrt{u_{th,j}^n}$,
$\xi_{1,j}^n$ is the first entry of $\bxi_{j}^n$, and
$\bu_j^n,u_{th,j}^n$ are the mean macroscopic velocity and thermal
velocity in the $j$-th cell. Since $|\alpha + e_1| = M+1$ when
$|\alpha| = M$, $v_1f_{j}^n(\bv)$ no longer exists in the space
$F_{M}(\bu_{j}^n,u_{th,j}^n)$. By simply dropping the terms with
$|\alpha+e_1| = M+1 $, we project $v_1f_{j}^n(\bv)$ back into
$F_{M}(\bu_{j}^n,u_{th,j}^n)$, since when $|\alpha|> M$,
$H_{\alpha}(\bxi)$ is orthogonal to $F_{M}(\bu,u_{th})$ with respect
to the inner product
\begin{equation}
  \label{eq:inner-product}
  (f,g) = \int_{\bbR^3}f(\bxi)g(\bxi)\exp\left(\frac{|\bxi|^2}{2}\right)\mathrm{d}\bxi.
\end{equation}

The other point is how to add up two distribution functions lying in
$F_{M}(\bu_{j}^n,u_{th,j}^n)$ and $F_{M}(\bu_{j+1}^n,u_{th,j+1}^n)$
respectively. The proposition in \cite{NRxx} is referred to solve it.
\begin{proposition}
  Suppose $f\in F_{M}(\bu_1,u_{th,1})$ can be represented by
 \begin{equation}
   f(\bv) = \sum\limits_{|\alpha|\leqslant M}f_{1,\alpha}
\mathcal{H}_{u_{th,1},\alpha}(\bxi_1),\quad \bxi_1 = (\bv-\bu_1)/\sqrt{u_{th,1}}.
 \end{equation}
 For some $\bu_2\in \bbR^3$ and $u_{th,2} > 0$,
 $\{F_{\alpha}(\tau)\}_{|\alpha|\leqslant M}$ satisfies
\begin{equation}
  \left\{
    \begin{array}{lr}
      \dfrac{{\rm d} F_{\alpha}}{{\rm d} \tau} = \displaystyle\sum_{d=1}^3
      S^2 \left[ u_{th,1} RF_{\alpha-2e_d} + w_d\sqrt{u_{th,1}}F_{\alpha-e_d} 
      \right], & \quad \forall \tau \in [0,1], \\
      F_{\alpha}(0) = f_{1,\alpha},
    \end{array}
\right.
\end{equation}
where $S$ and $R$ are given below. And $\bw =
(\bu_1-\bu_2)/\sqrt{u_{th,2}}$, $\hat{u}_{th}=\sqrt{u_{th,1}/u_{th,2}}$. 
\begin{equation}
  R(\tau) = \frac{\hat{u}_{th} -1}{(\hat{u}_{th}-1)\tau +
    1}, \quad S(\tau) = 1- \tau R(\tau) =
  \frac{1}{(\hat{u}_{th} -1)\tau +1}.
\end{equation}
Let 
\begin{equation}
  g(\bv) = \sum\limits_{|\alpha|\leqslant
    M}F_{\alpha}(1)\mathcal{H}_{u_{th,2},\alpha}(\bxi_2),
  \quad  \bxi_2 = (\bv - \bu_2)/\sqrt{u_{th,2}}.
\end{equation}
Then $g(\bv) \in F_{M}(\bu_2,u_{th,2})$ and $g(\bv)$ satisfies 
\begin{equation}
\label{eq:inner_product}
  \int_{\bbR^3}p(\bv)f(\bv)\mathrm{d}\bv = \int_{\bbR^3}p(\bv)g(\bv)\mathrm{d}\bv, \quad
  \forall p(\bv) \in P_{M}(\bbR^3).
\end{equation}
\end{proposition}
The proposition provides an algorithm to project a function in
$F_{M}(\bu_1, u_{th,1})$ to the space $F_{M}(\bu_2,
u_{th,2})$. Assuming $f_1 \in F_{M}(\bu_1,u_{th,1})$ and $f_2 \in
F_{M}(\bu_2,u_{th,2})$, we can find an $\tilde{f}_1\in
F_M(\bu_2,u_{th,2})$ as a representation of $f_1 \in
F_M(\bu_1,u_{th,1})$ in the sense of \eqref{eq:inner_product}.  Thus,
to add up $f_j\in F_{M}(\bu_{j}^n, u_{th,j}^n)$ and $f_{j+1}\in
F_{M}(\bu_{j+1}^n, u_{th,j+1}^n)$, we first find an $\hat{f}_{j+1}\in
F_{M}(\bu_{j}^n, u_{th,j}^n)$ as an approximation of $f_{j+1}$ in the
sense of \eqref{eq:inner_product}, and then add up the coefficients of
$f_{j}$ and $\hat{f}_{j+1}$ for the same basis function.
 
The regularization term appears only in the moment equations
containing $\partial{f_{\alpha}}/\partial{t}$ with $|\alpha| =
M$. Therefore, only when $|\alpha| = M$, we have to calculate
$K_{2,j}^n$. We simply use the central difference scheme to
approximate the spatial derivatives in \eqref{RM}:
\begin{equation}
\begin{split}
  \pd{u_d}{x}& \approx \nabla_x^h u_{d,j}^n \triangleq \frac{u_{d,j +
      1}^n - u_{d,j-1}^n}{2\Delta x}, \\
  \pd{u_{th}}{x}& \approx \nabla_x^h u_{th,j}^n \triangleq
  \frac{u_{th,j + 1}^n - u_{th,j-1}^n}{2\Delta x}.
\end{split}
\end{equation} 
Then we get the numerical approximation for $K_{2,j}^n$ with $|\alpha| =
M$, as
\begin{equation}
\begin{split}
\label{revision_scheme}
K_{2,j}^n = -\Delta t\sum_{|\alpha|=M}
\left(\alpha_1+1\right)\sum\limits_{d=1}^3\left(f_{\alpha-e_d+e_1}^{n}
  \nabla_x^h u_{d,j}^n + \frac{f_{\alpha-2e_d+e_1}^{n}}{2} \nabla_x^h
  u_{th,j}^n \right) \mathcal{H}_{u_{th,j}^n,\alpha}
\left(\frac{\bv-\bu_j^n}{\sqrt{u_{th,j}^n}}\right).
\end{split}
\end{equation}

\subsection{The acceleration and collision step}\label{collision_step}
The acceleration step is performed by solving
\begin{equation}\label{eq:true_force}
  \pd{u_1}{t} - F_{1} = 0.
\end{equation}
For the time step $\Delta t $, \eqref{eq:true_force} is approximated
as
\begin{equation} \label{eq:force_solve}
   u^{n+1}_{1,j} = u^{n+1,\ast}_{1,j} 
  + \Delta t F_{1,j}^{n+1},
\end{equation}
where $u_{1,j}^{n+1,\ast}$ is the first entry of the macroscopic
velocity $\bu$ in the $j$-th cell after the convection step at $t =
t^n$. And $F_{1,j}^{n+1}$ is the first entry of the electric force
$\bF$ in the $j$-th cell after the convection step at $t = t^n$. In
the 1D spatial space, the Poisson equation \eqref{eq:force} reduces
into a second order ODE as
\begin{equation}\label{eq:one-dim}
  F_1(t,x,\bv) = \frac{q}{m}E_1(t,x), \quad E_1(t,x) =
  -\pd{\psi(t,x)}{x}, \quad -\partial_{xx}\psi = \frac{q \rho}{\epsilon_0}.
\end{equation}
The three-point central difference scheme is used to discretize the
potential equation
\begin{equation} \label{eq:difference_rho} 
  - \frac{\psi_{j+1}^{n+1} - 2\psi_j^{n+1} + \psi_{j-1}^{n+1}}{\Delta x^2} 
  = \frac{q \rho_j^{n+1}}{\epsilon_0},
\end{equation}
where $\rho_j^{n+1}$ is the density in the $j$-th cell after the
convection step, noticing that the density is not updated in the
acceleration step and the collision step. The central difference is
used to approximate $E_{1,j}^{n+1}$ and $F_{1,j}^{n+1}$
\begin{equation}\label{eq:E}
E_{1,j}^{n+1} = - \frac{\psi_{j+1}^{n+1} - \psi_{j-1}^{n+1}}{2\Delta x},
\quad F_{1,j}^{n+1} = \frac{q}{m} E_{1,j}^{n+1}.
\end{equation}

For the BGK collision model, the moment expansion results in a simple
form \eqref{eq:collision}, and  \eqref{eq:collision_part} changes into
\begin{equation} \label{eq:collision_form}
  \pd{f_{\alpha}}{t} = -\nu f_{\alpha}, \quad |\alpha|\geqslant 2. 
\end{equation} 
It is directly integrated as
\begin{equation}\label{eq:collision_form_1}
  f_{j,\alpha}^{n+1} = f_{j,\alpha}^{n+1,\ast}\exp(-\nu
     \Delta t^n), \quad \forall \alpha\in \bbN^3, \quad 1 <
     |\alpha|\leqslant M.
\end{equation}
\begin{remark}
  The collision step only revises the coefficients with
  $2\leqslant|\alpha|\leqslant M$, and does not change the macroscopic
  velocity $\bu$ and the density $\rho$, which are decided by the
  coefficients with $|\alpha|<2$.
\end{remark}

\subsection{Outline of the algorithm}
The overall numerical scheme is summarized as below: 
{\it\begin{enumerate}
\item Let $n = 0$ and set the initial value of $f_{j,\alpha}^{n}$;
\item Calculate $\Delta t^n$ according to the CFL condition;
\item Integrate the convection term using \eqref{eq:scheme};
\item Update the macroscopic velocity using \eqref{eq:E} and \eqref{eq:force_solve};
\item Apply the collision step \eqref{eq:collision_form_1};
\item Let $n \leftarrow n+1$, and return to step 2.
\end{enumerate}}

\begin{remark}
  For the V-P equation, $\nu = 0$, and the collision step is simply
  skipped over.
\end{remark}
The scheme keeps both the discretized mass and momentum
conserved. Precisely, in the case of the periodic boundary condition, let us
denote the discretized mass as
\begin{equation}
  \mathcal{D}_h(t_n) = \sum_{j=1}^N \Delta x \int_{\bv\in\bbR^3} f_j^n(\bv) \mathrm{d} \bv
\end{equation}
and discretized momentum as
\begin{equation}
  \mathcal{M}_h(t_n) = \sum_{j=1}^N \Delta x \int_{\bv\in\bbR^3} \bv f_{j}^{n}\mathrm{d}\bv,
\end{equation}
we have the following conclusion.
\begin{theorem}
\label{thm:mu_conservation}
The numerical solution $f_h^n(x, \bv)$ given by the algorithm above in
the case of the periodic boundary condition satisfies that
\begin{equation}
  \mathcal{D}_h(t_n) = \mathcal{D}_h(t_0), \quad
  \mathcal{M}_h(t_n) = \mathcal{M}_h(t_0)
\end{equation}
for all $n > 0$.
\end{theorem}
\begin{proof}
  Noticing that the mass on each cell is not modified when we  apply the
  regularization term and in the acceleration step, the conservation
  of the mass is straight forward based on results in \cite{Cai}.

  The momentum conservation is equivalent to verify
  \begin{equation}\label{eq:mu-con}
    \mathcal{M}_h(t_{n+1}) = \mathcal{M}_h(t_n).
  \end{equation}
  It is clear that the collision step does not change the macroscopic
  velocity and the density. Therefore we verify below that the momentum is
  conserved in the acceleration step and the  convection step, respectively.
  \begin{enumerate}
  \item We first verify that the acceleration step keeps momentum
    conservation. From \eqref{eq:pri} and \eqref{eq:force_solve}, we
    have
    \begin{equation}
      \begin{split}
        \label{eq:mu-con-1}
        \mathcal{M}_h(t_{n+1})& = \sum_{j=1}^N\Delta
        x\int_{\bv\in\bbR^3}\bv f_{j}^{n+1}\mathrm{d}\bv  \\
        & = \Delta x
        \sum_{j=1}^N \rho_j^{n+1} \bu_{j}^{n+1} \\
        & = \Delta x \sum_{j=1}^N \rho_j^{n+1}
        (\bu_{j}^{n+1,\ast} +  \bF_{j}^{n+1}\Delta t^n)   \\
        & =\Delta x \sum_{j=1}^N \rho_j^{n+1} \bu_j^{n+1,\ast}
        +\Delta x\Delta t^n \sum_{j=1}^N
        \bF_{j}^{n+1}  \rho_j^{n+1}.    
\end{split}
\end{equation}
According to \eqref{eq:difference_rho} and \eqref{eq:E}, we have
\begin{equation}
\begin{split}
\label{eq:final_mu}
\sum_{j=1}^N F_{1,j}^{n+1} \rho_j^{n+1} &=
\sum\limits_{j=1}^N\left(-\frac{\epsilon_0}{q}\right)\frac{\psi_{j+1} - 2\psi_j +
  \psi_{j-1}}{\Delta x^2}\left(-\frac{q}{m}\right)\frac{\psi_{j+1} -
  \psi_{j-1}}{2\Delta x} \\
&= \dfrac{\epsilon_0}{2m\Delta x^3} \sum\limits_{j=1}^N
\big(\psi_{j+1}\psi_{j+1} -
\psi_{j-1}\psi_{j-1} -2\psi_j\psi_{j+1}  \\
&\qquad\qquad\qquad +2\psi_j\psi_{j-1} + \psi_{j-1}\psi_{j+1} -
\psi_{j+1}\psi_{j-1}\big) \\
& = \psi_{N}^2 + \psi_{N+1}^2 - \psi_0^2-\psi_1^2 + \psi_0\psi_1 -
\psi_N\psi_{N+1}.
\end{split}
\end{equation} 
Since we restrict the problem to the 1D spatial space with the
periodic boundary condition, we deduce by $\psi_{N+1} = \psi_1$,
$\psi_{0} = \psi_{N}$, $F_2 \equiv 0$ and $F_3 \equiv 0$ that
\[\sum_{j=1}^N \bF_{j}^{n+1} \rho_j^{n+1} = 0,\]
and we obtain
\begin{equation}
  \label{eq:force_con_1}
  \mathcal{M}_h(t_{n+1}) = \Delta x \sum_{j=1}^N \rho_j^{n+1} \bu_j^{n+1,\ast}.
\end{equation}
Since $\rho_j^{n+1}$ is the density after the $n$-th convection step,
$\Delta x \sum_{j=1}^N \rho_j^{n+1} \bu_j^{n+1,\ast}$ is the total
momentum after the $n$-th convection step.

\item Here we verify that the convection step does not change the
  total momentum. Thanks to \eqref{eq:pri} and \eqref{eq:scheme}, we
  have
  \begin{equation}
    \begin{split}
      & \Delta x \sum_{j=1}^N \rho_j^{n+1} \bu_j^{n+1,\ast} 
      = \sum\limits_{j =1}^N \Delta x\int_{\bv\in\bbR^3} \bv
      f_{j}^{n+1,\ast}\mathrm{d}\bv\\
      = & \sum\limits_{j=1}^N\Delta x \int_{\bv\in\bbR^3} \bv[f_j^n
      + K_{1,j}^n + K_{2,j}^n]\mathrm{d}\bv\\
      = & \sum\limits_{j=1}^N\Delta x \int_{\bv\in\bbR^3} \bv
      f_j^n\mathrm{d}\bv + \sum\limits_{j=1}^N\Delta x
      \int_{\bv\in\bbR^3} \bv K_{1,j}^n\mathrm{d}\bv
      +\sum\limits_{j=1}^N\Delta x
      \int_{\bv\in\bbR^3} \bv K_{2,j}^n\mathrm{d}\bv\\
      = & \sum\limits_{j =1}^N \Delta x\int_{\bv\in\bbR^3} \bv
      f_{j}^{n}\mathrm{d}\bv - \Delta t^n\sum\limits_{j
        =1}^N\int_{\bv\in\bbR^3} \bv (F_{j+1/2}^n -
      F_{j-1/2}^n)\mathrm{d}\bv
      + \sum\limits_{j=1}^N\Delta x 
      \int_{\bv\in\bbR^3} \bv K_{2,j}^n\mathrm{d}\bv \\
      = & \mathcal{M}_h(t_{n}) - \Delta t^n\int_{\bv\in\bbR^3} \bv
      (F_{N+1/2}^n - F_{1/2}^n)\mathrm{d}\bv +
      \sum\limits_{j=1}^N\Delta x \int_{\bv\in\bbR^3} \bv
      K_{2,j}^n\mathrm{d}\bv .
     \end{split}
   \end{equation}
   Due to the periodic boundary condition, we have that $F_{N+1/2}^n =
   F_{1/2}^n$. Meanwhile, the regularization part only updates the
   $M$-th order terms which have no effect on the macroscopic velocity
   $\bu$ and the density $\rho$, thus the regularization term will not
   break the momentum conservation. Precisely speaking, the basis
   functions of $K_{2,j}$ are $\mathcal{H}_{u_{th,j}^n,\alpha} \left(
     (\bv-\bu_j^n)\Big/\sqrt{u_{th,j}^n} \right)$, with $|\alpha|=M$,
   $M \geqslant 3$, which are orthogonal to $\bv$, then
   \begin{equation}
   \int_{\bv\in\bbR^3} \bv K_{2,j}^n\mathrm{d}\bv = 0,
   \end{equation} 
   and we obtain
   \begin{equation}
    \label{eq:convection_con_1}
     \Delta x \sum_{j=1}^N \rho_j^{n+1} \bu_j^{n+1,\ast} =  \mathcal{M}_h(t_{n}).
   \end{equation}
  \end{enumerate}
  With \eqref{eq:force_con_1} and \eqref{eq:convection_con_1}, we
  conclude the total momentum conservation consequently
\begin{equation}
  \mathcal{M}_h(t_{n+1}) =  \mathcal{M}_h(t_{n}).
\end{equation} 
This ends the proof.
\end{proof}
 

%% file: article_num_ex.tex
\section{Numerical Examples}

We study the linear Landau damping modelled by  the V-P and V-B equations
with the periodic boundary conditions. The CFL number is always set as
$0.45$. The specific examples are from \cite{Filbet}. The form of the
V-B equations with the periodic boundary conditions coupled with the
normalized Poisson equation is
\begin{gather} \label{eq:normal_force} 
  \pd{f}{t} + \bv\cdot \nabla_{x}f + E \cdot \nabla_{\bv} f
  = \nu(f_M - f),  \\
  E(t,x) = - \nabla_x\psi(t,x), \\
  -\Delta \psi(t,x) = 
  \int_{\bbR^3} f(t,x,\bv)\mathrm{d}\bv -1.
\end{gather}
Here we adopt the same  initial data  as in \cite{Filbet} 
\begin{equation}
\label{eq:initial_data}
f(0,x,\bv) = f_M = \frac{1}{\sqrt{2\pi}}e^{-|\bv|^2/2}(1+A
   \cos(kx)), \quad  (x,\bv) \in (0,L)\times \bbR^3,
\end{equation}
where $A$ is the amptitude of the perturbation, $k$ denotes the wave
number, and the periodic length is taken as $L = 2\pi/k$. The initial
data and \eqref{eq:expansion_1} give us
\begin{equation}
  \rho(0,x) =  f_{0}(0,x) = (1+A \cos(kx)), \quad  \left.f_{\alpha}\right|_{t=0} = 0,~~
  ~|\alpha|>0.
\end{equation}

What of one's interests is the evolution of the square root of the
electric energy, which is defined as
\begin{equation} \label{eq:energy} 
\mathcal{E}_h(t) = \sum_{j=1}^N \Delta x E_j^2(t).
\end{equation}
According to Landau's theory, the time evolution of the square root of
$\mathcal{E}_h(t)$ is expected to be exponentially decaying almost
with a fixed rate $\gamma_L$, which is affected by the wave number $k$
and the collision frequency $\nu$. For this purpose, we always plot
the square root of $\mathcal{E}_{h}(t)$ in logarithm scale in the
figures in this section. The total energy is the sum of the electric
energy and the kinetic and internal energy of the particles as
\begin{equation}\label{eq:total_energy}
  \mathcal{E}_{total}(t) = \mathcal{E}_h (t) + \mathcal{E}_p (t),
\end{equation}
where the kinetic and internal energy of the particles
\begin{equation}
  \mathcal{E}_p (t) = \Delta x \sum_{i=1}^N \left(\rho_i(t) u_i^2(t) 
    + \rho_i(t) u_{th,i}(t) \right).
\end{equation}

\subsection{Numerical resolution study}
\label{sec:convergence}

\begin{figure}[!ht]
\centering
\subfigure[$k=0.3$]{
\begin{overpic}[scale=.4]{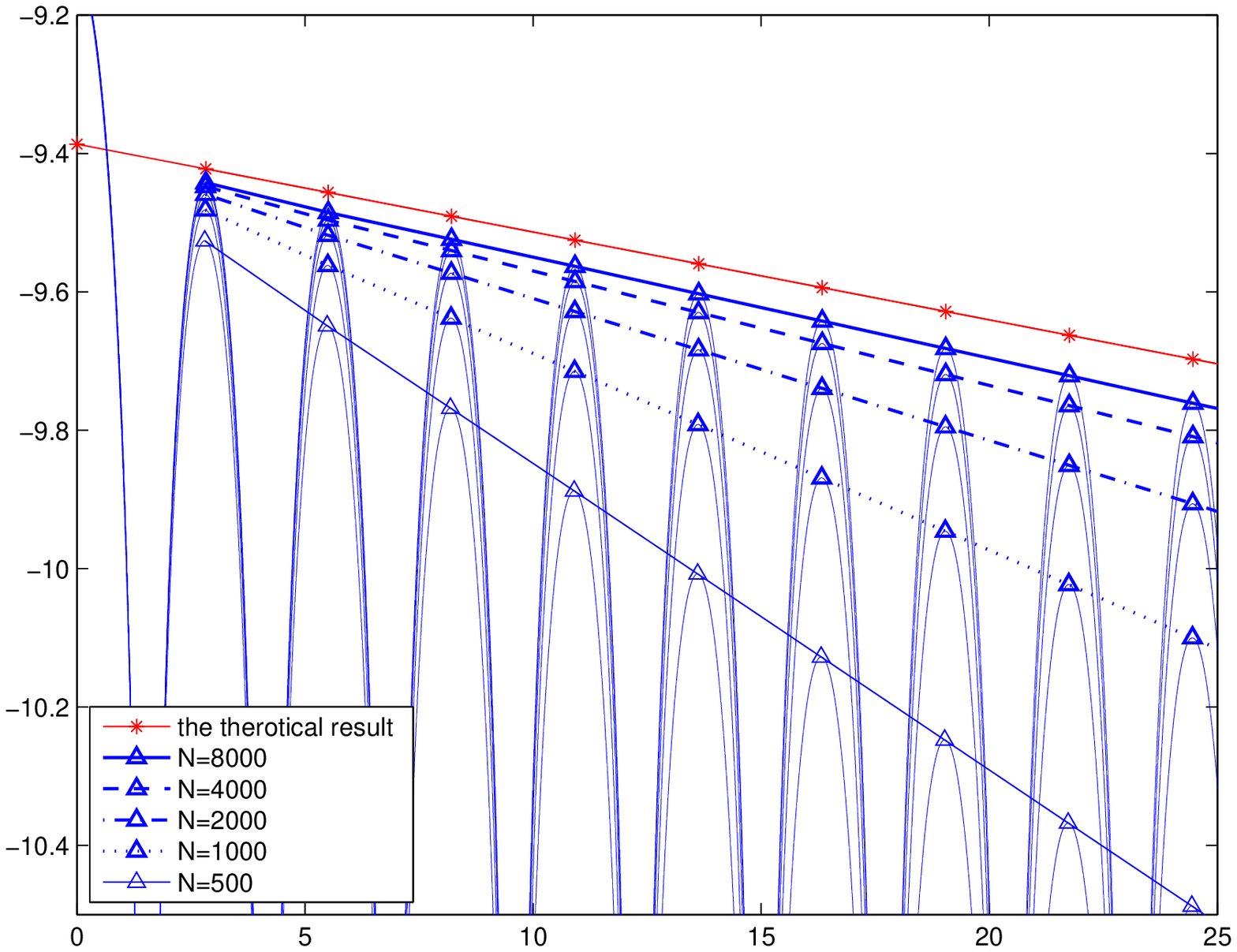}
\end{overpic}
}
\subfigure[$k=0.5$]{
\begin{overpic}[scale=.4]{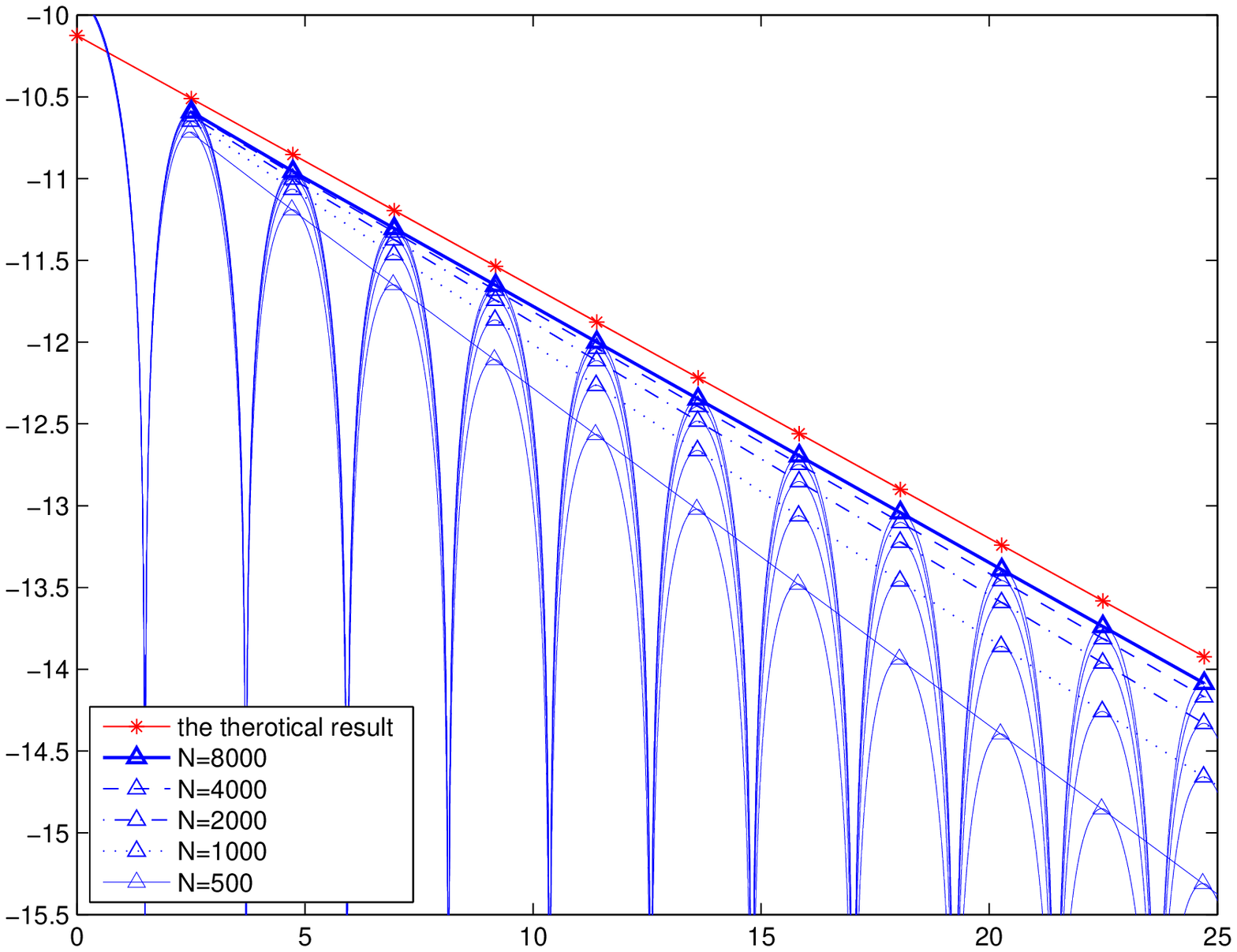}
\end{overpic}
}
\caption{Exponentially damping in time of the square root of
  $\mathcal{E}_h$ on different spatial grids with the wave number $k =
  0.3$ and $0.5$. The curves in blue are the square root of
  $\mathcal{E}_h$ in time using logarithm scale on different spatial
  grid size. The slopes of the blue lines are the numerical damping
  rate $\gamma_L^h$ by the least square fitting of the peak value
  points of $\mathcal{E}_h$. The slope of the red line is the damping
  rate given by the theoretic data in Table
  \ref{tab:limit_damping_rate_collisionless}.}
\label{fig:spatial}
\end{figure}

\begin{figure}[!ht]
\centering
\subfigure[$k=0.3$]{
\begin{overpic}[scale = .4]{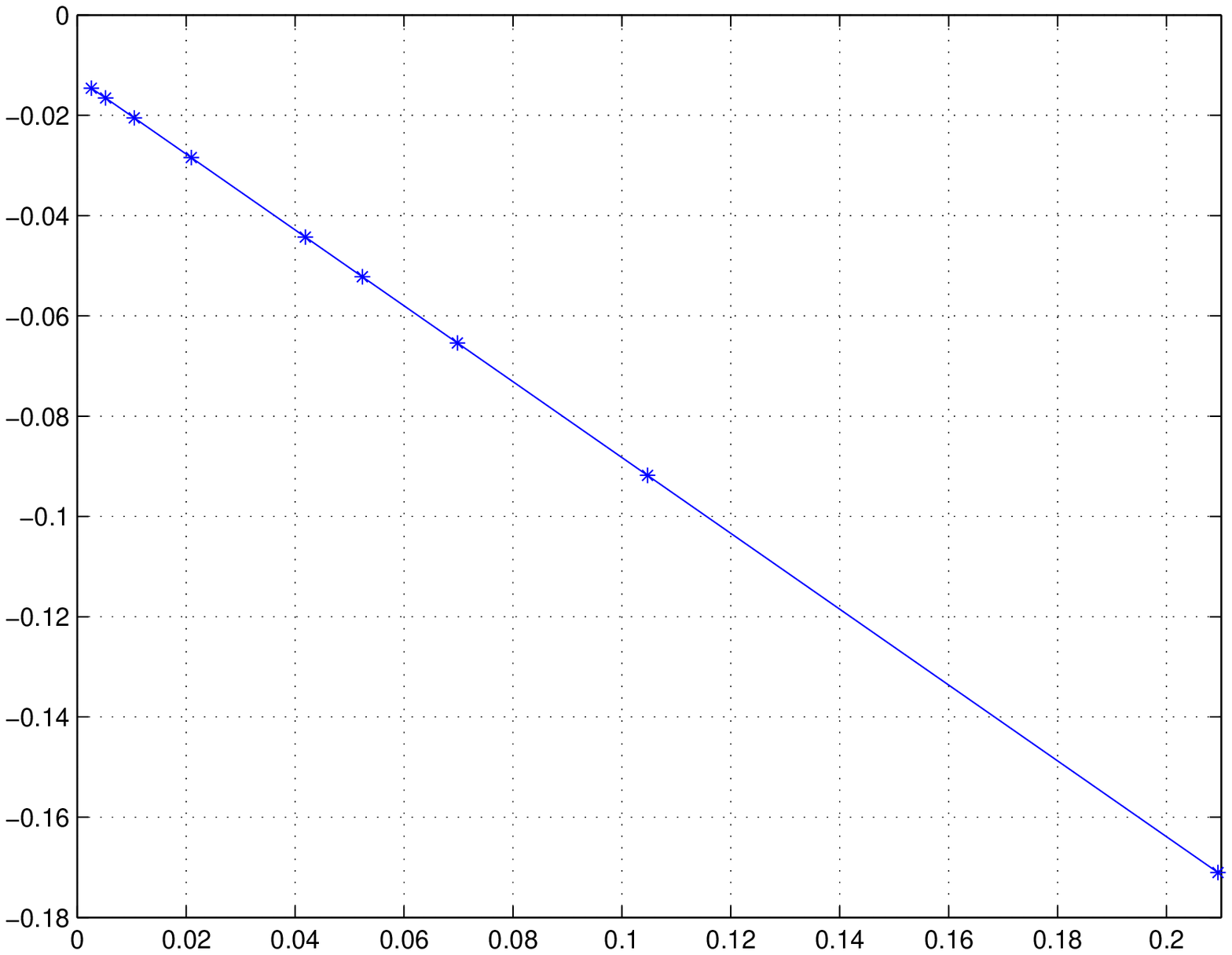}
\end{overpic}
}
\subfigure[$k=0.5$]{
\begin{overpic}[scale = .4]{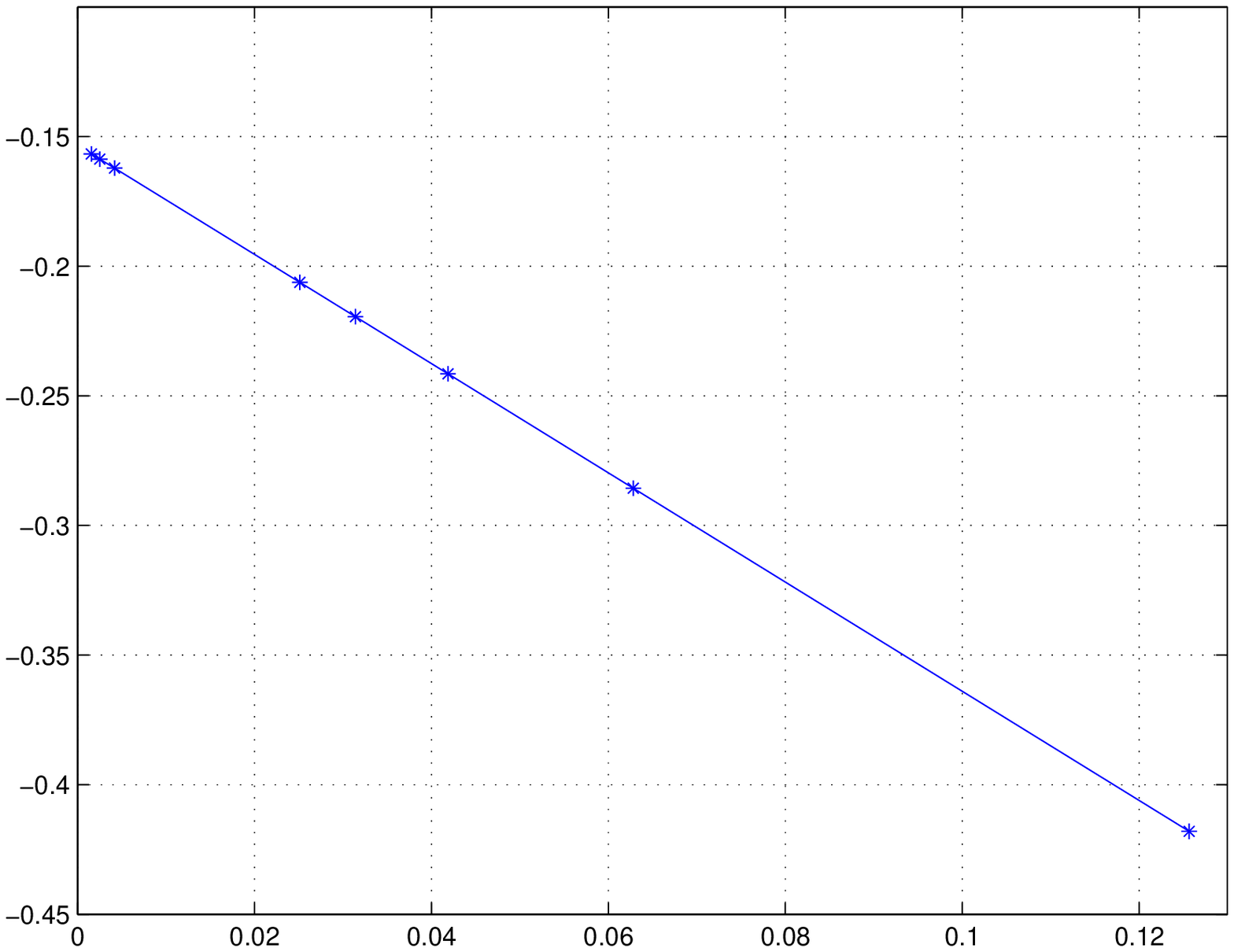}
\end{overpic}
}
\caption{The linear dependence of the numerical damping rates in
  spatial grid size with the wave number $k=0.3$ and $0.5$. The
  $x$-axis is the spatial grid size $\Delta x$ and the $y$-axis is
  numerical damping rate $\gamma_L^h(\Delta x)$. The line is obtained
  by the least square fitting of $\gamma_L^h(\Delta x)$ with $\Delta
  x$ ranging from $L/100$ to $L/8000$. The intercept of the line on
  $y$-axis is the parameter $\gamma_L^{h,0}$ and the slope of the line
  is the parameter $\gamma_L^{h,1}$.}
\label{fig:slope}
\end{figure}

We first examine the numerical convergence on different spatial grid
size. The V-P equations in 1D spatial space and 1D velocity space are
numerically solved. The number of moments is set as $80$. In Figure
\ref{fig:spatial}, the evolution of the square root of
$\mathcal{E}_h(t)$ with the wave numbers $k = 0.3$, $0.5$ on different
grid size is presented. The number of spatial grids we used are $500$,
$1000$, $2000$, $4000$ and $8000$. It is clear that the square root of
$\mathcal{E}_h(t)$ is damping exponentially on all different grid
size. With the increasing of the grid number, the damping rate is
decreasing monotonically. For the spatial grid size $\Delta x$, we
adopt the least square fitting, which uses the peak value points of
$\mathcal{E}_h(t)$, to obtain the numerical damping rate $\gamma_L^h$.
One finds obviously in Figure \ref{fig:spatial} that both
$\mathcal{E}_h(t)$ and the damping rate $\gamma_L^h$ are converging
while the spatial grid size $\Delta x$ is going to zero. Furthermore,
if we take the numerical damping rate $\gamma_L^h$ as a function of
the spatial grid size $\Delta x$ and apply a least square fitting to
retrieve the parameters $\gamma_L^{h,0}$ and $\gamma_L^{h,1}$ (see
Figure \ref{fig:slope}) in the ansatz as
\[ 
\gamma_L^h(\Delta x) = \gamma_L^{h,0} + \gamma_L^{h,1} \Delta x, 
\]
it is found that the fitting provides us both parameters extremely
close to constants for all the values of $k$ we tested ranging from
$0.2$ to $0.5$. This indicates us that the following relation 
\begin{equation}\label{extrapolate_gamma}
\gamma_L^h(\Delta x) - \gamma_L^{h,0} \propto \Delta x
\end{equation}
is approximately valid. The obtained parameter $\gamma_L^{h,0}$ is
regarded as the limit of the numerical damping rate when $\Delta x$ is
going to zero. To our surprise, the limit $\gamma_L^{h,0}$ we obtained
is in perfect agreement with the theoretic data in Table
\ref{tab:limit_damping_rate_collisionless}.

\begin{table}[!ht]
  \centering
  \begin{tabular}[!ht]{|c|c|c|} \hline 
    Wave number $k$ & Theoretic data \cite{Eric} & $\gamma_L^{h,0}$ \\ \hline 
    $0.2$ &   $-5.5\times10^{-5}$     & $-9.28\times10^{-5}$ \\ \hline 
    $0.3$ &   $-0.0126$               & $-0.01260$  \\ \hline
    $0.4$ &   $-0.0661$               & $-0.06614$  \\ \hline
    $0.5$ &   $-0.1533$               & $-0.15334$ \\ \hline   
  \end{tabular}
  \caption{The comparison of the limit numerical damping rates and the theoretic data.}
  \label{tab:limit_damping_rate_collisionless}
\end{table}

\begin{figure}[!ht]
\centering
\subfigure[$k=0.3$]{
\begin{overpic}[scale=.4]{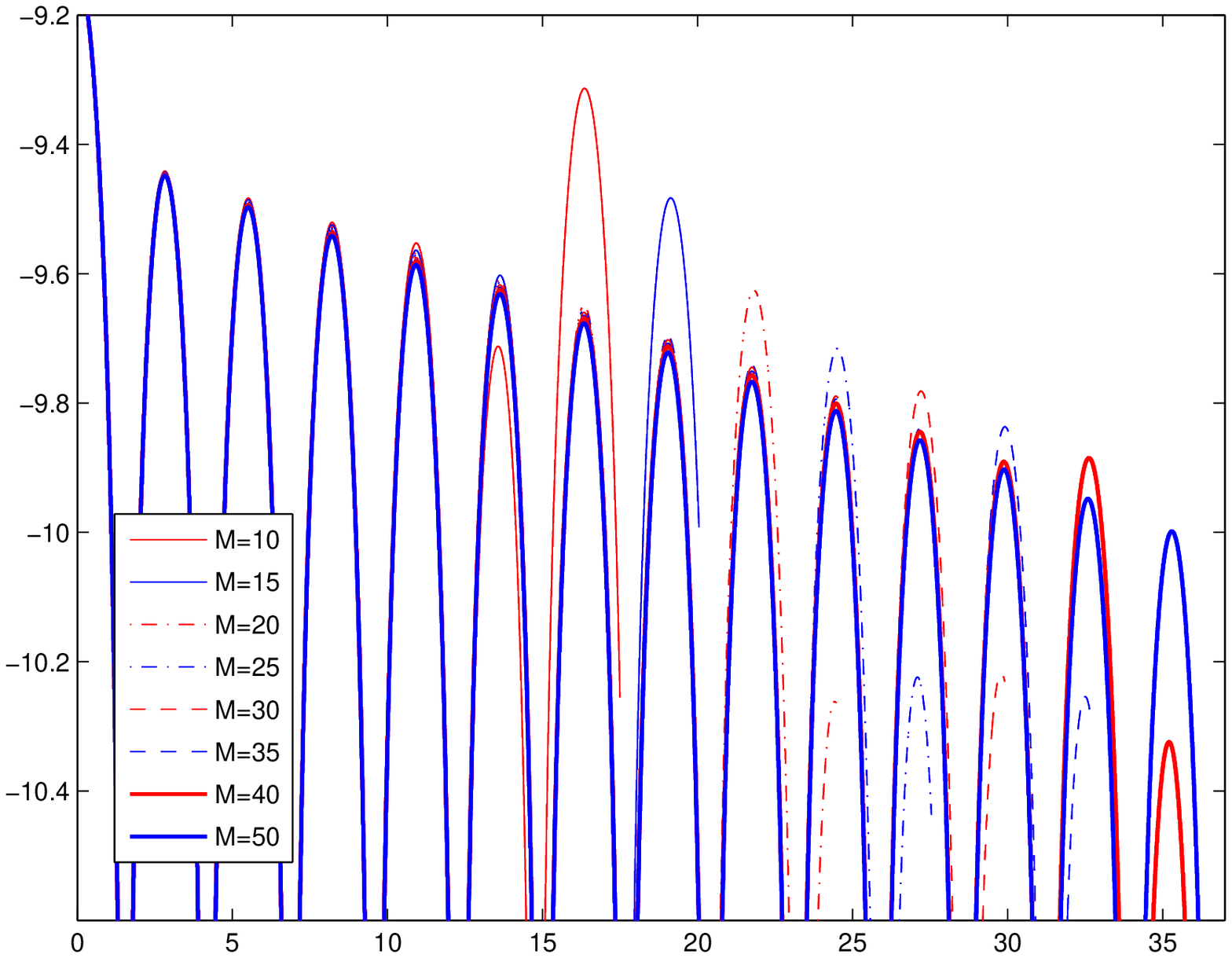}
\end{overpic}
}
\subfigure[$k=0.5$]{
\begin{overpic}[scale=.4]{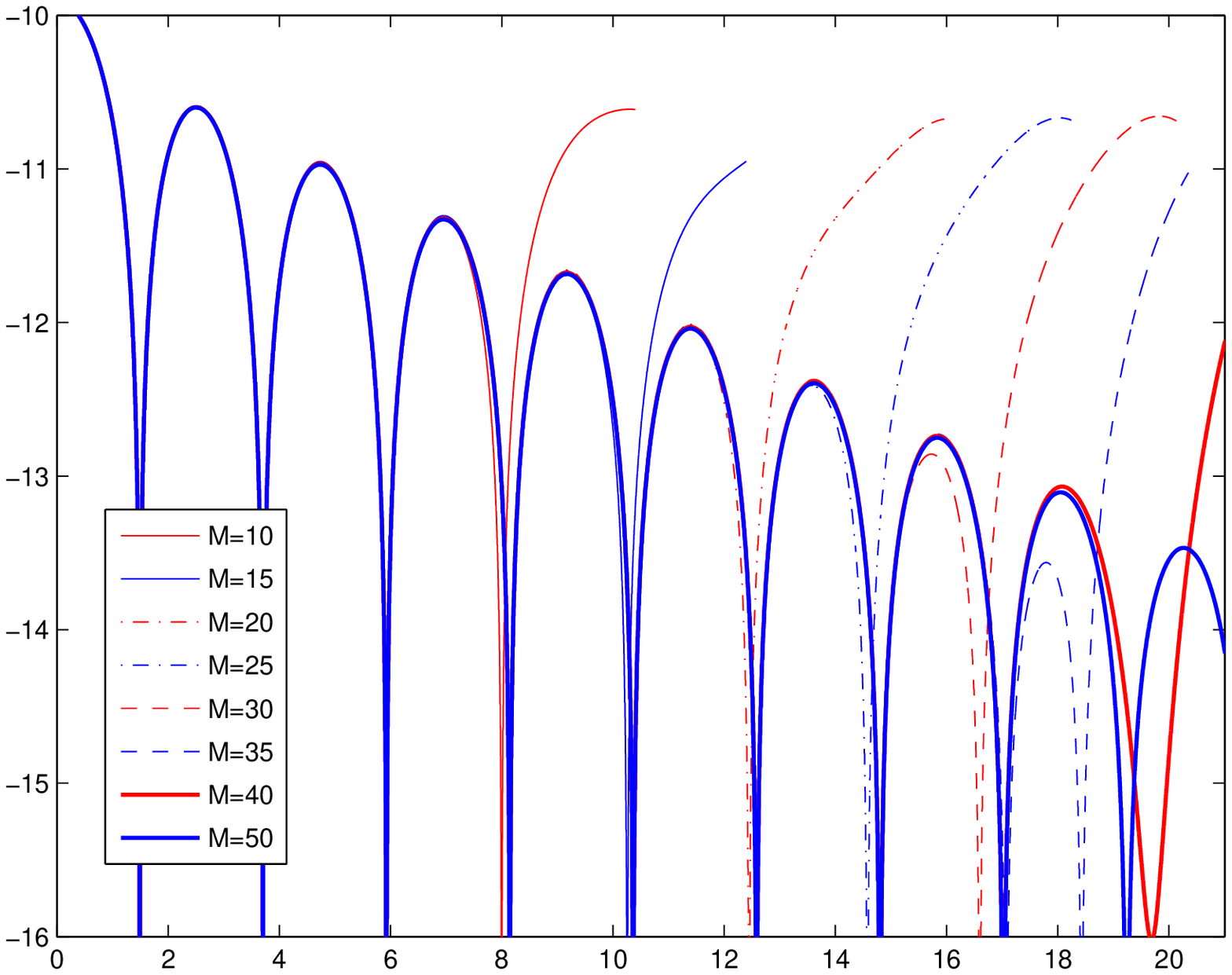}
\end{overpic}
}
\caption{Exponentially damping in time of the square root of
  $\mathcal{E}_h$ on different number of moments $M$ with the wave
  number $k=0.3$ and $0.5$. The curves are the square root of
  $\mathcal{E}_h$ in time using logarithm scale using different number
  of moments.}
\label{fig:moment}
\end{figure}

\begin{figure}[!ht]
  \centering
\subfigure[$k=0.3$]{
\begin{overpic}[scale=.4]{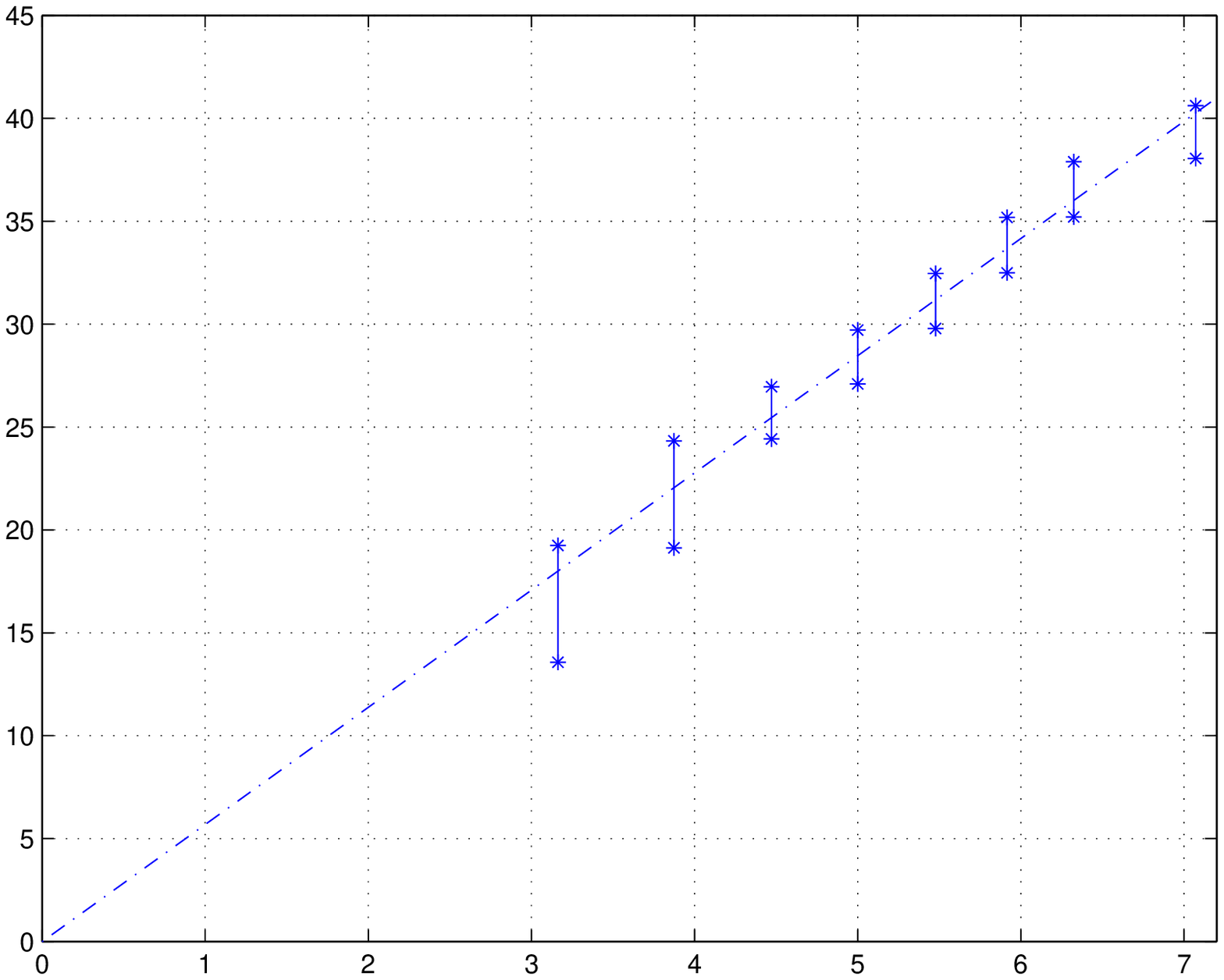}
\end{overpic}
}
\subfigure[$k=0.5$]{
\begin{overpic}[scale=.4]{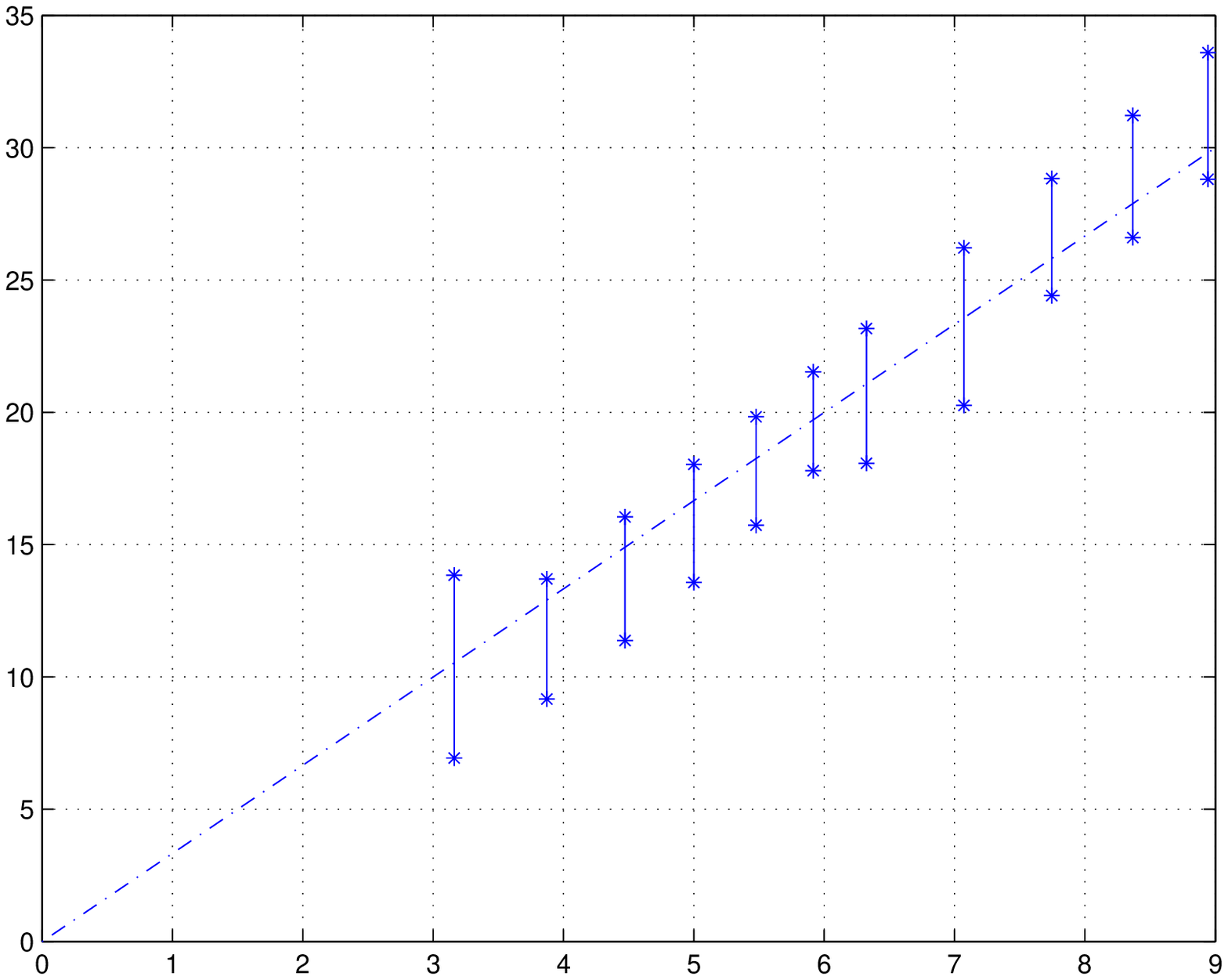}
\end{overpic}
}
\caption{The dependence of the recurrence time on the number of
  moments used. The $x$-axis is the square root of the number of
  moments $M$, and the $y$-axis is time. The recurrence time for a
  given $M$ is between the time of the sequential peak values of
  $\mathcal{E}_h$ before and after it deviates from the exponentially
  decay (see Figure \ref{fig:moment}). In the figures, the vertical
  line segments are given by two time of the sequential peak values.}
  \label{fig:recurrence}
\end{figure}

\begin{figure}[!ht]
\centering
\subfigure[$k=0.3$]{
\begin{overpic}[scale = .4]{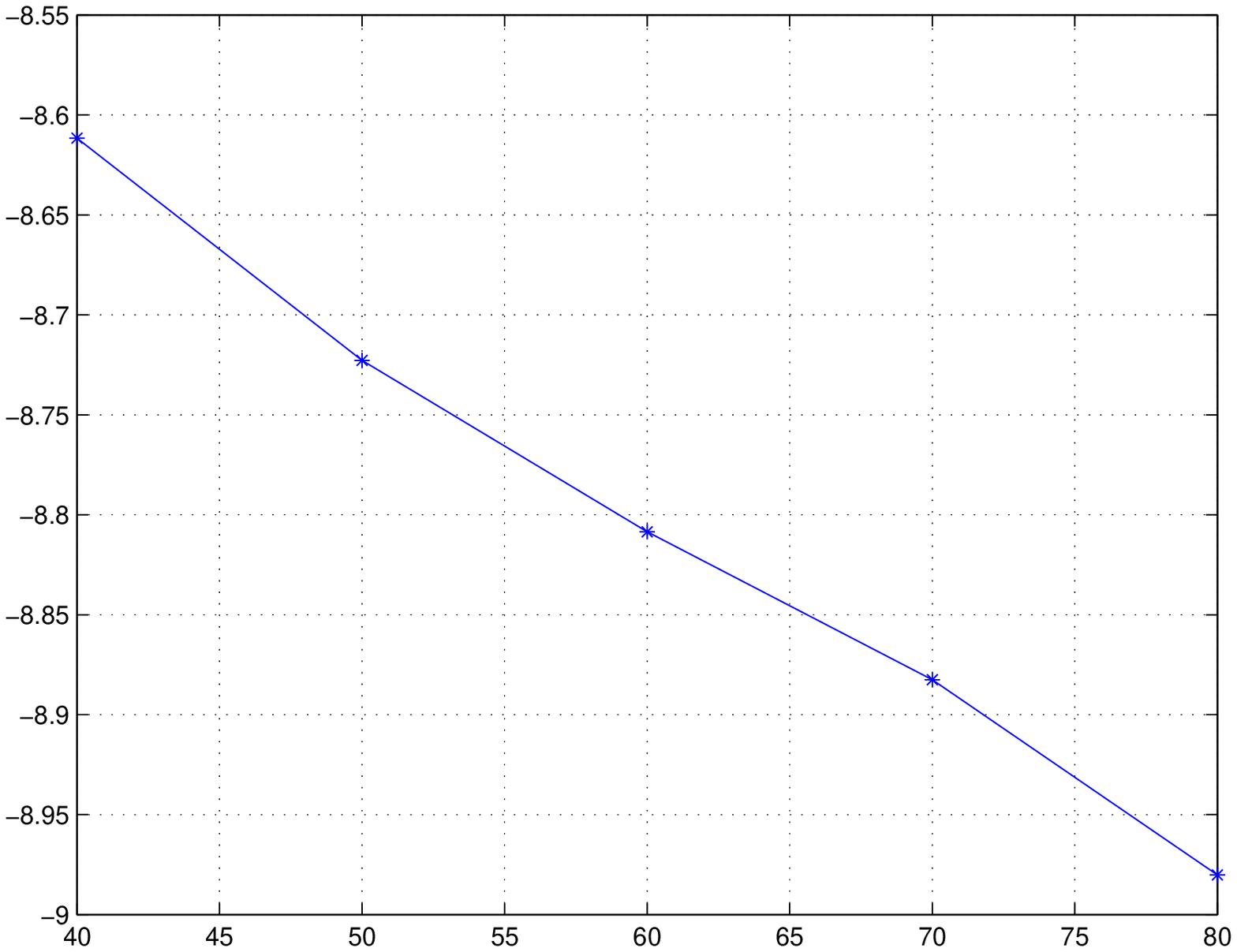}
\end{overpic}
}
\subfigure[$k=0.5$]{
\begin{overpic}[scale = .4]{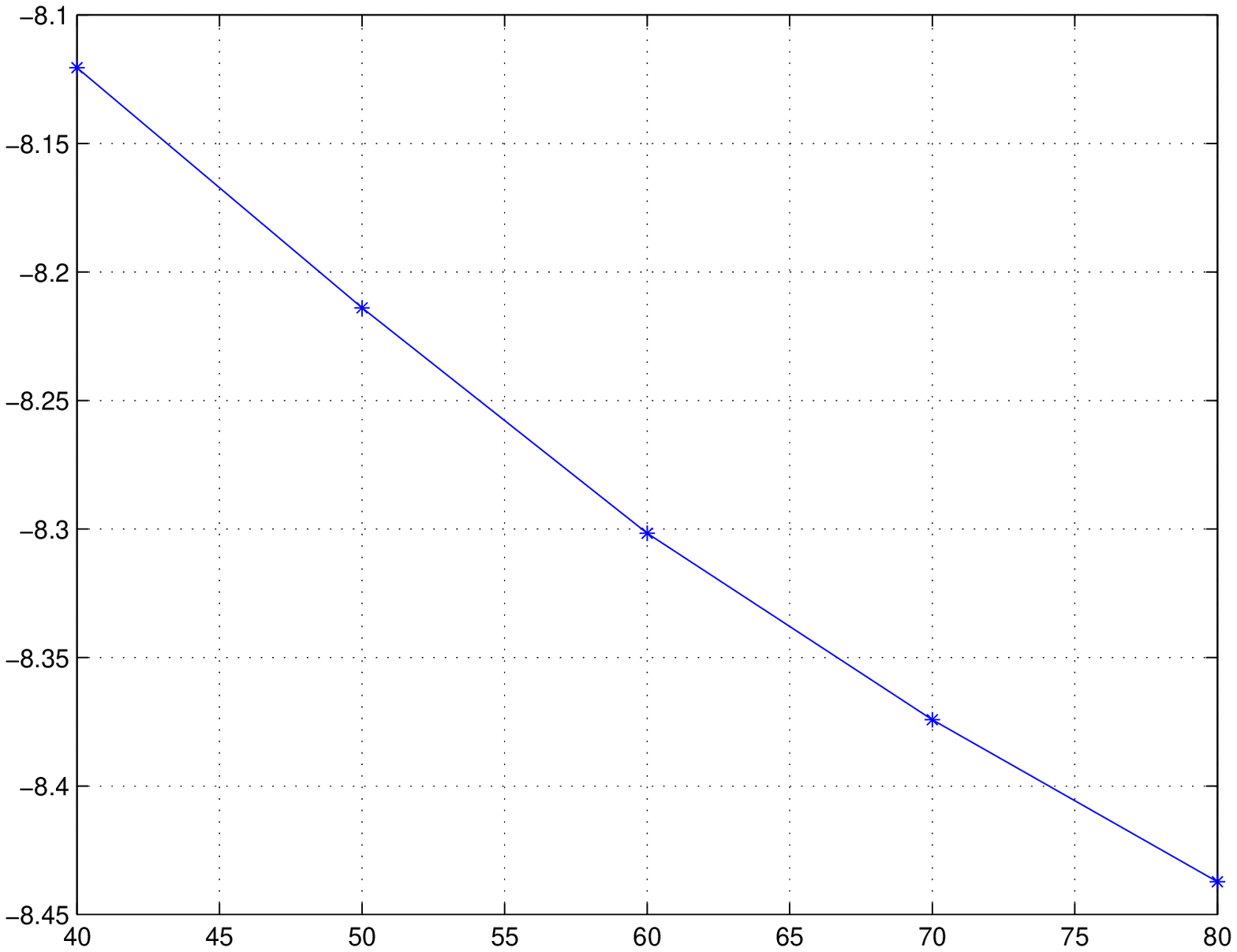}
\end{overpic}
}
\caption{The approximately exponential convergence of the numerical
  damping rates in the number of moments with the wave number $k=0.3$
  and $0.5$. The $x$-axis is the number of moments $M$, with $M_{i} =
  40, 50, 60, 70, 80, 90$ and $\Delta M = 10$. The $y$-axis is
  $\log(\gamma_L^h(M_{i}) - \gamma_L^h(M_{i-1}))$, where
  $\gamma_L^h(M_{i})$ is obtained by the least square fitting of the
  peak values of $\mathcal{E}_h$. The spatial grid size is fixed as
  $\Delta x = L/8000$.}
\label{fig:slope_mnt}
\end{figure}

Let us turn to the study of the numerical convergence in term of the
number of moments. Again the V-P equations in 1D spatial space and 1D
velocity space are numerically solved. The results using different
number of moments ranging from $10$ to $50$ are collected in Figure
\ref{fig:moment} with fixed number of spatial grids. In this figure,
we observe that the behavior of the time evolution of
$\mathcal{E}_h(t)$ for different number of moments are almost the same
at the beginning, indicating that the numerical damping rates are very
close to each other. The damping of the square root of
$\mathcal{E}_h(t)$  is persisting until the recurrence appears. After the
appearance of the recurrence, the evolving of $\mathcal{E}_h(t)$
deviates evidently from exponentially decaying.

The exponential convergence rate in the number of moments is expected
since the Hermite spectral expansion is used in the velocity space. We
notice that the exponential convergence in term of number of moments
is quite different to be observed, since the numerical error is in
linear convergence due to the dominance of the spatial discretization.
As the best try, a very fine spatial grid with $8000$ points, which is
about the maximal capacity of our current hardware, is used in the
computation to suppress the spatial discretization error. With the
fixed spatial grid size, we then take the numerical damping rate
$\gamma_L^h$ as a function of the number of moments $M$. In the case of an
exponential convergence, the dependence of $\gamma_L^h(M)$ on $M$ is
as
\begin{equation}\label{eq:ansatz_mnt}
\gamma_L^h(M) = \gamma_L^{m,0} + \gamma_L^{m,1} \lambda^{-M},
\end{equation} 
where $\gamma_L^{m,0}$, $\gamma_L^{m,1}$ and $\lambda$ are parameters
independent of $M$. Let $M_i$ be a given arithmetic sequence with a
const $\Delta M = M_{i} - M_{i-1}$. Based on the ansatz
\eqref{eq:ansatz_mnt}, we have that
\[
\begin{array}{rcl}
\gamma_L^h(M_i) &=& \gamma_L^{m,0} + \gamma_L^{m,1} \lambda^{-M_i}, \\ [2mm]
\gamma_L^h(M_{i-1}) &=& \gamma_L^{m,0} + \gamma_L^{m,1} \lambda^{-M_{i-1}}.
\end{array}
\]
By substracting the two equations and then taking a logarithm on both
sides, we have
\[
\log \left(\gamma_L^h(M_i) - \gamma_L^h(M_{i-1}) \right) = -M_i
\log(\lambda) + \log \gamma_L^{m,1} (1 - \lambda^{\Delta M}).
\]  
It is to find that $\log \left(\gamma_L^h(M_i) - \gamma_L^h(M_{i-1})
\right)$ is approximately linear in $M_i$ if \eqref{eq:ansatz_mnt} is
valid. We give the plot of $\log \left(\gamma_L^h(M_i) -
  \gamma_L^h(M_{i-1}) \right)$ in $M_i$ for $M_i = 40$, $50$, $60$,
$70$, $80$, and $90$ in Figure \ref{fig:slope_mnt}. It is clear in
this figure that $\log \left(\gamma_L^h(M_i) - \gamma_L^h(M_{i-1})
\right)$ is almost linearly related with $M_i$, indicating that
\eqref{eq:ansatz_mnt} is valid and the numerical convergence rate in
term of the number of moments is approximately exponential.

It has been pointed out in \cite{Eric} that the recurrence is an
essential phenomenon of a class of numerical methods for the Vlasov
equations. Such a phenomenon is mainly the effect of the free
streaming part of the Vlasov equation. In the 1D case, it is
\begin{equation}
\pd{f}{t} + v \pd{f}{x} = 0.
\end{equation}
It is not difficult to find that for any velocity $v_i$, we have $f(j
L / v_i, x, v_i) = f(0, x, v_i)$ for an arbitrary integer $j$ when the
periodic boundary condition is imposed.  Therefore, if a numerical
scheme approximates the distribution function by taking its values on
$v_0, \cdots, v_M$ in the velocity space, then, for a time $T$ such
that most of $T / (L / v_i)$, $i = 0,\cdots,M$ are close to some
positive integers, the discrete distribution function $f(T, x, v)$
will be close to the initial setting $f(0, x, v)$. Thus the recurrence
occurs.

Our method is not able to escape away from this trap, either, since as
pointed out in \cite{Fan}, the hyperbolic moment equation is similar
as a discrete velocity model with a shifted and scaled stencil. Based
on our observation in Figure \ref{fig:recurrence}, it is conjectured
that the recurrence time is proportional to the square root of the
moment expansion order $M$. For a finite difference discretisation in
the velocity space, the recurrence time is proportional to the grid
size in $\bv$ \cite{Eric}. Noticing that the minimal distance between
the zeros of $M$-th degree Hermite polynomial is around $\sqrt{M}$,
the linear dependence of the recurrence time on $\sqrt{M}$ consists
with the analysis in \cite{Eric} if we regard the moment method as a
collocation spectral method in the velocity space with the collocation
points being the zeros of Hermite polynomial.

It has been proved in Theorem 1 that the mass and the momentum are
conserved by our scheme, while the total energy is not conserved. To
examine the behavior of the total energy of our method, we present in
Figure \ref{fig:energy} the variation of the total energy in time in
serveral different setups. It is clear that the total energy of the
numerical solution produced by our method is changed very slightly in
the whole computation. 

\begin{figure}[!ht]
\centering
\subfigure[$k=0.2$]{
\begin{overpic}[width=.45\textwidth,,height=5cm]{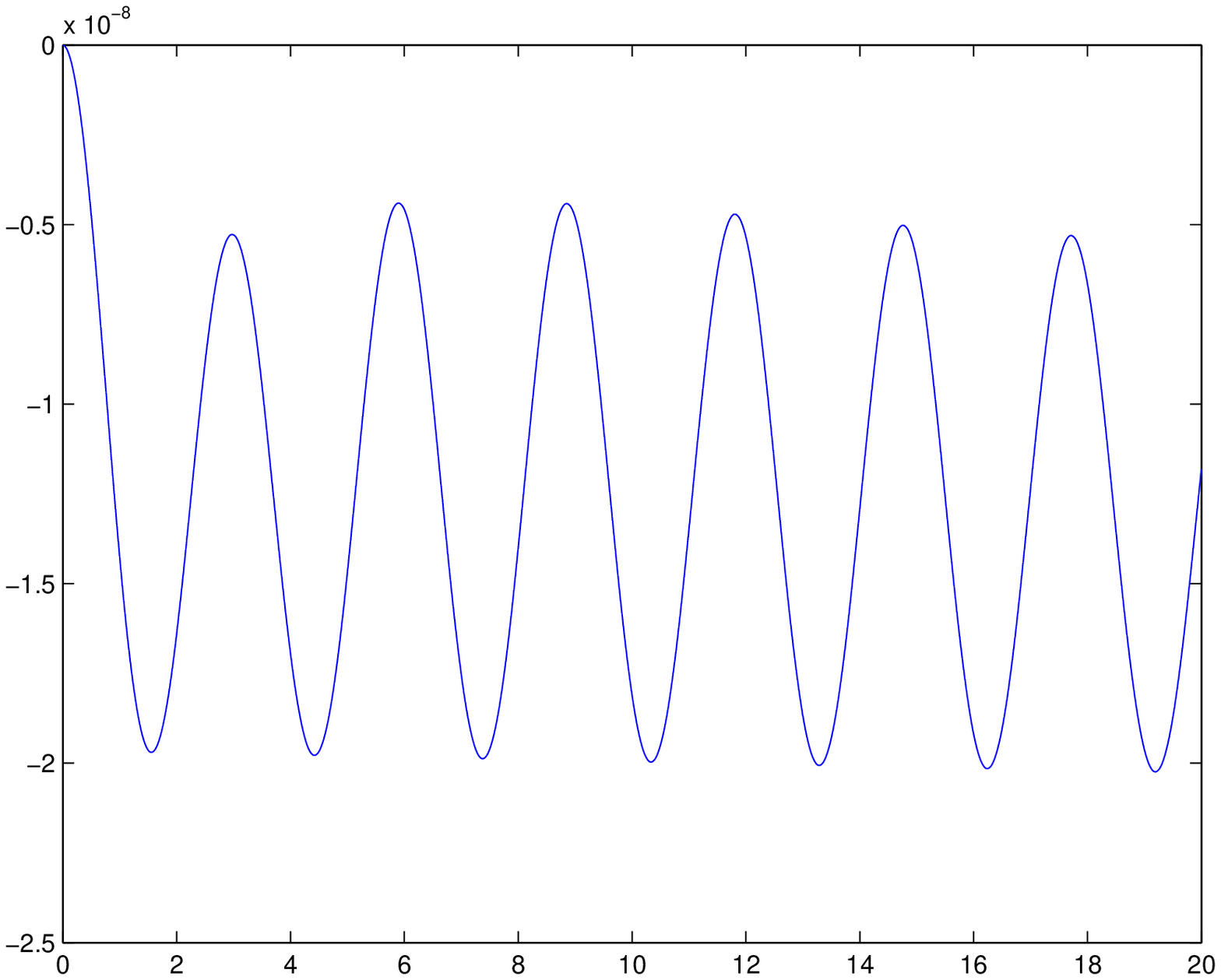}
\end{overpic}
}
\subfigure[$k=0.3$]{
\begin{overpic}[width=.45\textwidth,,height=5cm]{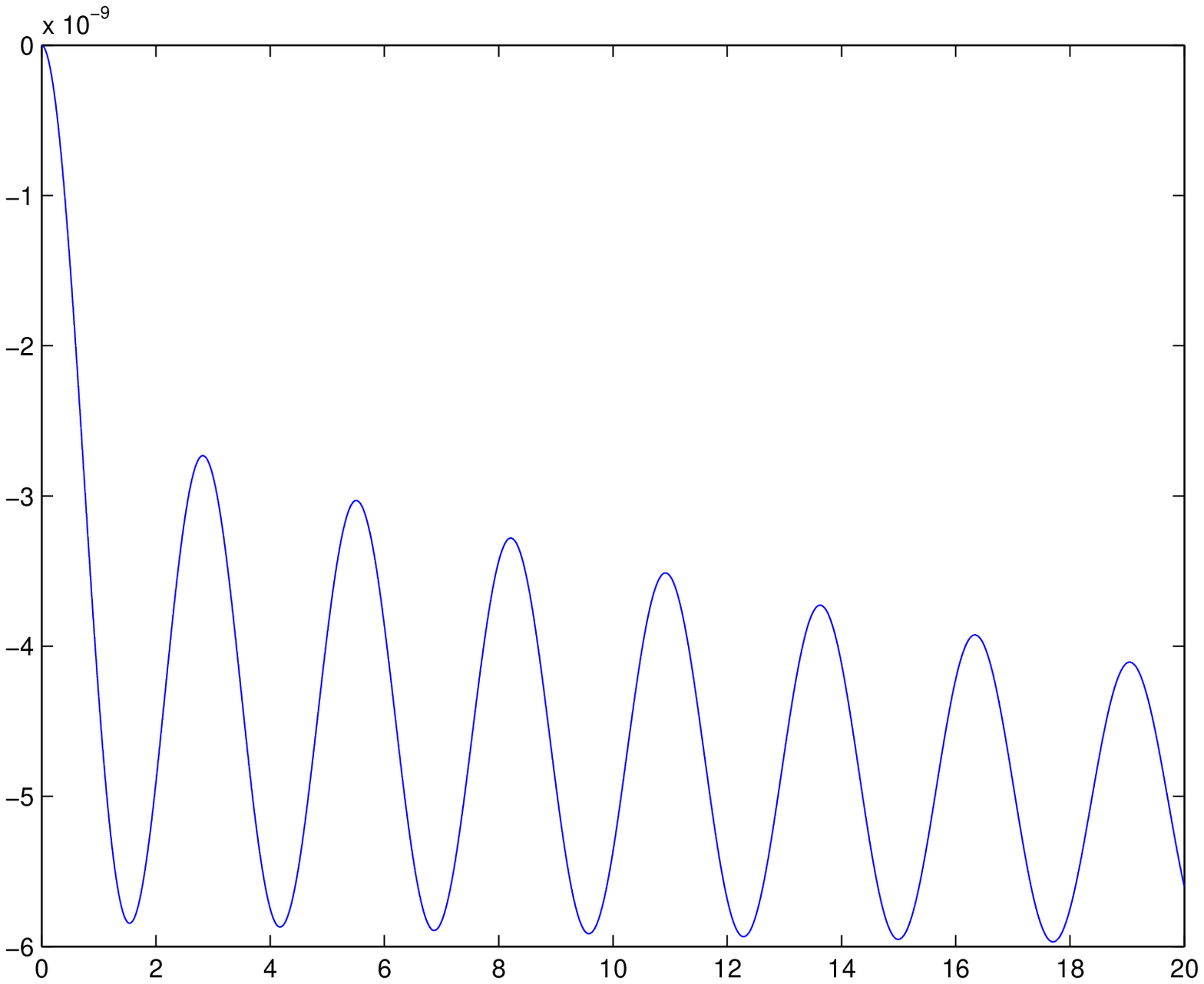}
\end{overpic}
}\\
\subfigure[$k=0.4$]{
\begin{overpic}[width=.45\textwidth,,height=5cm]{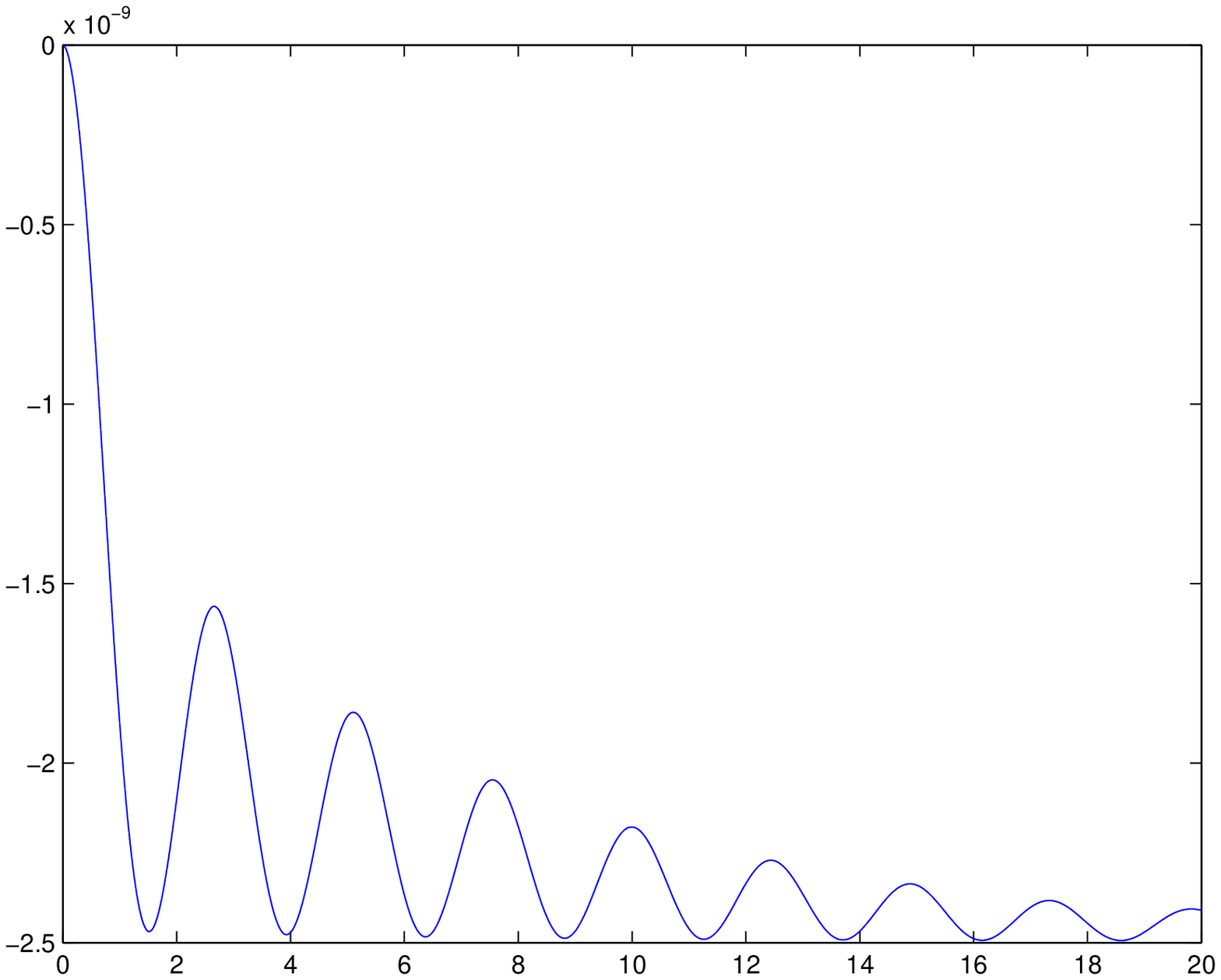}
\end{overpic}
}
\subfigure[$k=0.5$]{
\begin{overpic}[width=.45\textwidth,height=5cm]{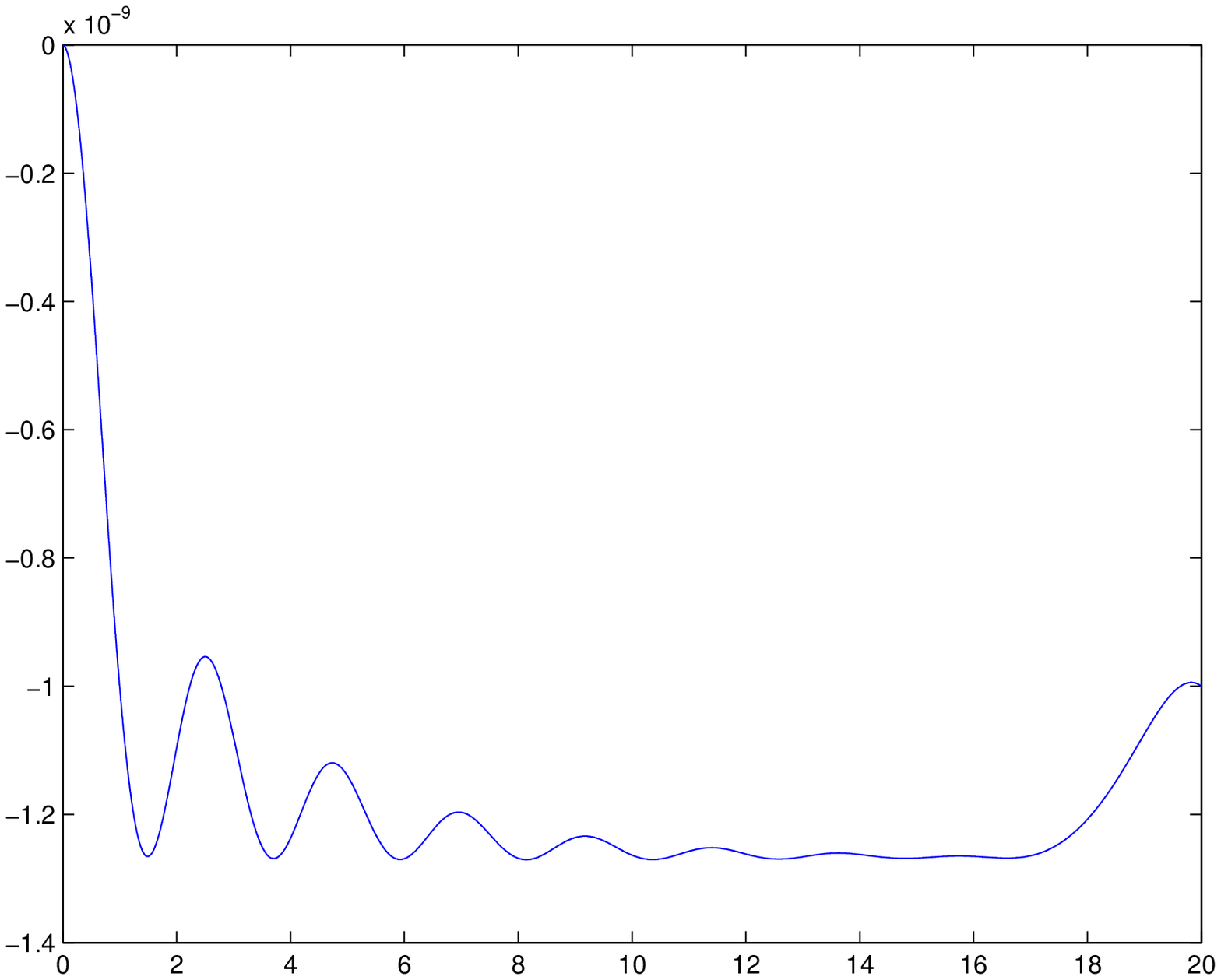}
\end{overpic}
}
\caption{The variation of the total energy $\mathcal{E}_{total}$ in
  time in the collisionless case with the wave number $k=0.2$, $0.3$,
  $0.4$ and $0.5$. The $x$-axis is time and the $y$-axis is
  $\mathcal{E}_{total}(t) - \mathcal{E}_{total}(0)$. Noticing that
  $\mathcal{E}_{total}(0)$ is of $O(1)$, it is clear  that  the
  variation is very small, though the total
  energy is not conserved.}
\label{fig:energy}
\end{figure} 

\subsection{Parameter study of linear Landau damping} 

\begin{figure}[!ht] \centering 
    \subfigure[$\nu = 0,k=0.2$]{
    \begin{overpic}[width=.45\textwidth,height=5cm]{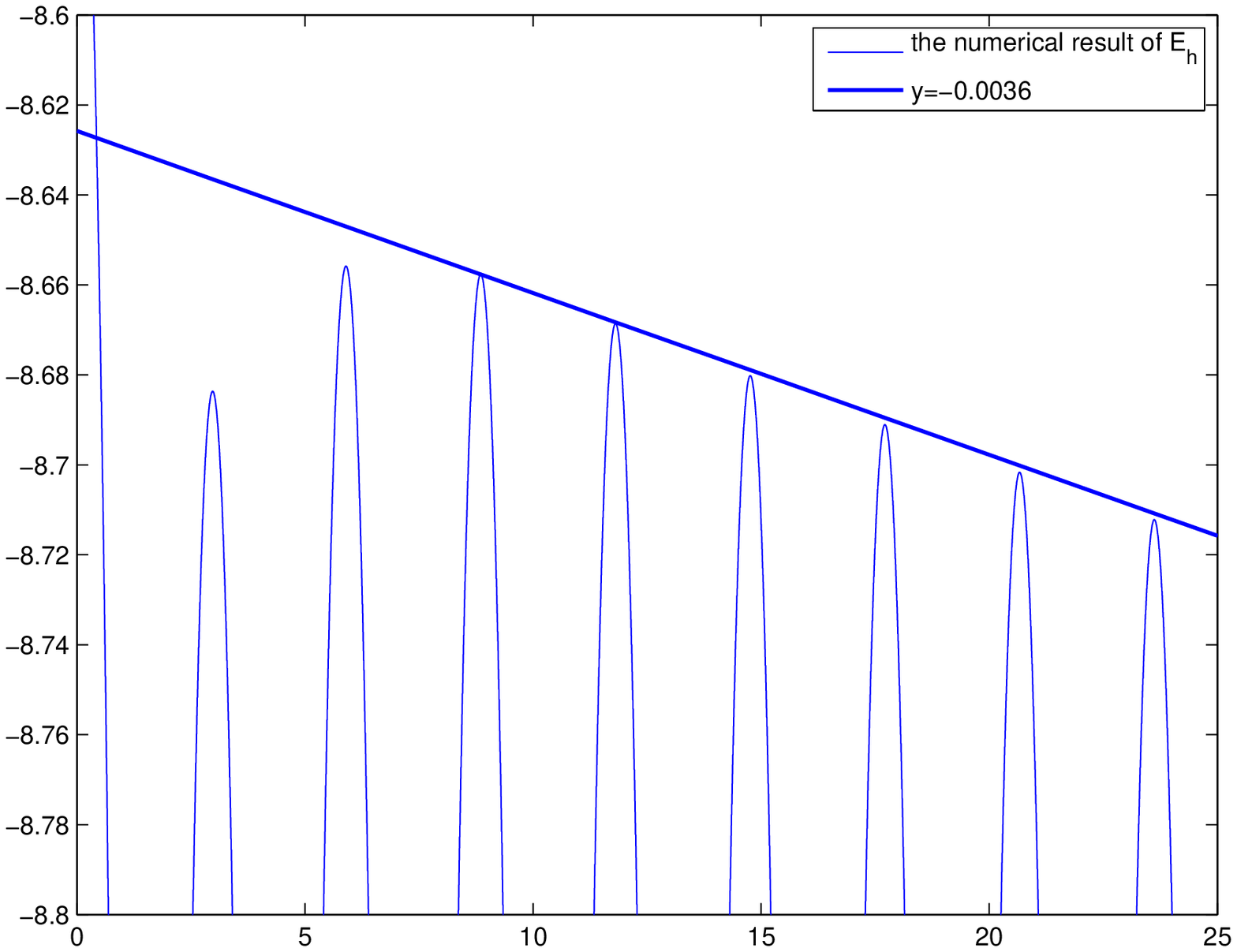}
    \end{overpic} } 
\subfigure[$\nu = 0,k=0.3$]{
    \begin{overpic}[width=.45\textwidth,height=5cm]{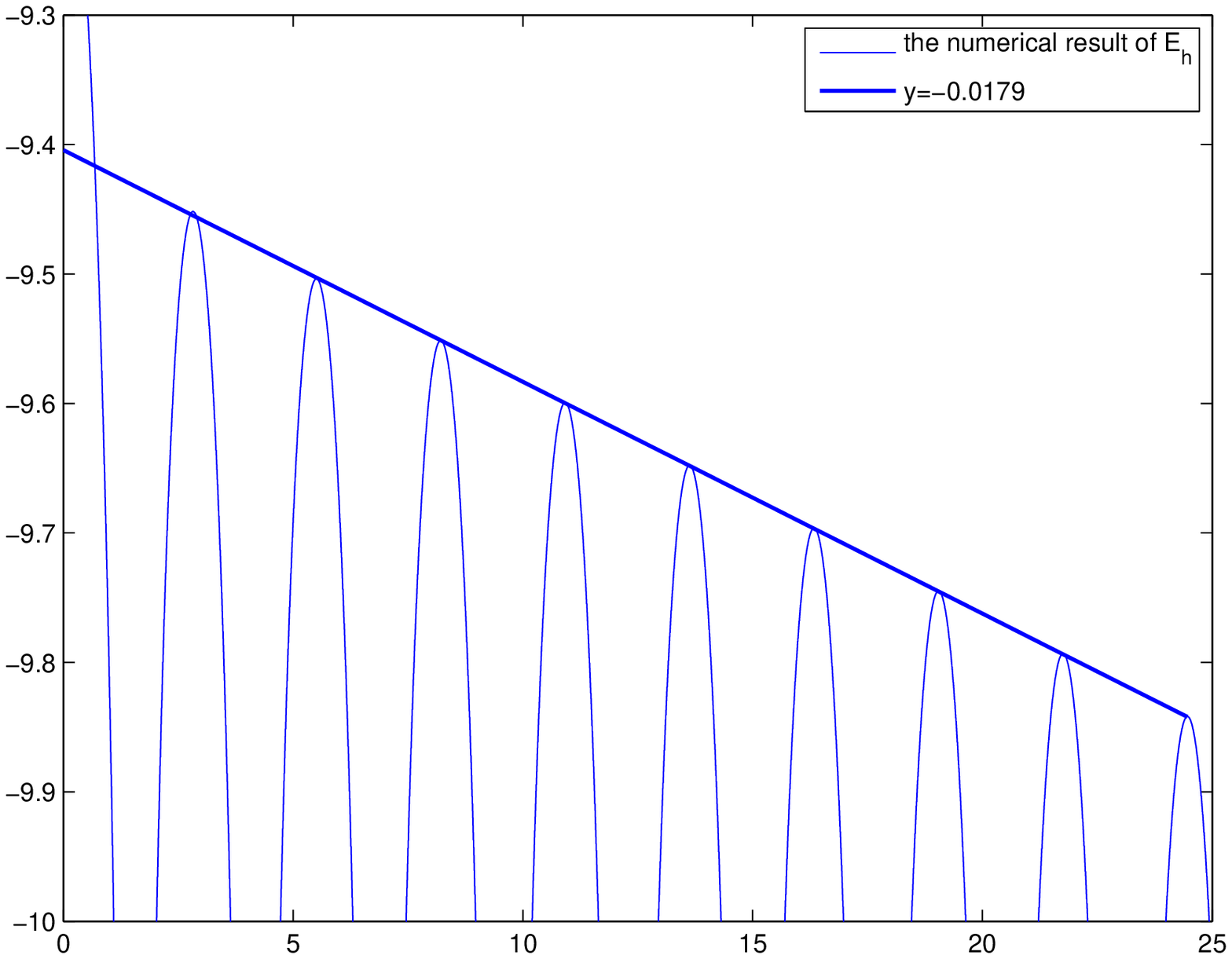}
    \end{overpic} }\\
 \subfigure[$\nu=0,k=0.4$]{
    \begin{overpic}[width=.45\textwidth,height=5cm]{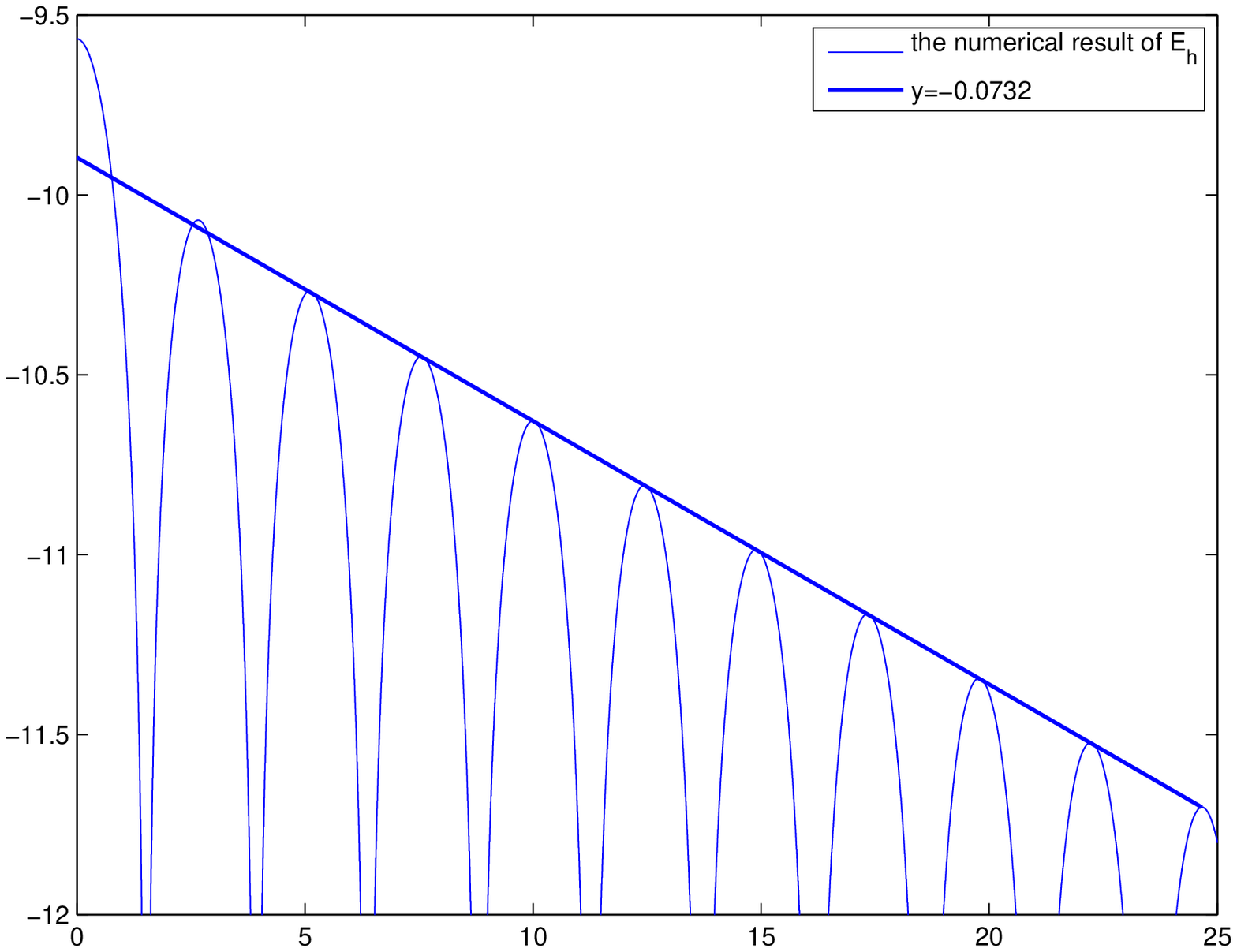}
    \end{overpic} } 
\subfigure[$\nu=0,k=0.5$]{
  \begin{overpic}[width=.45\textwidth,height=5cm]{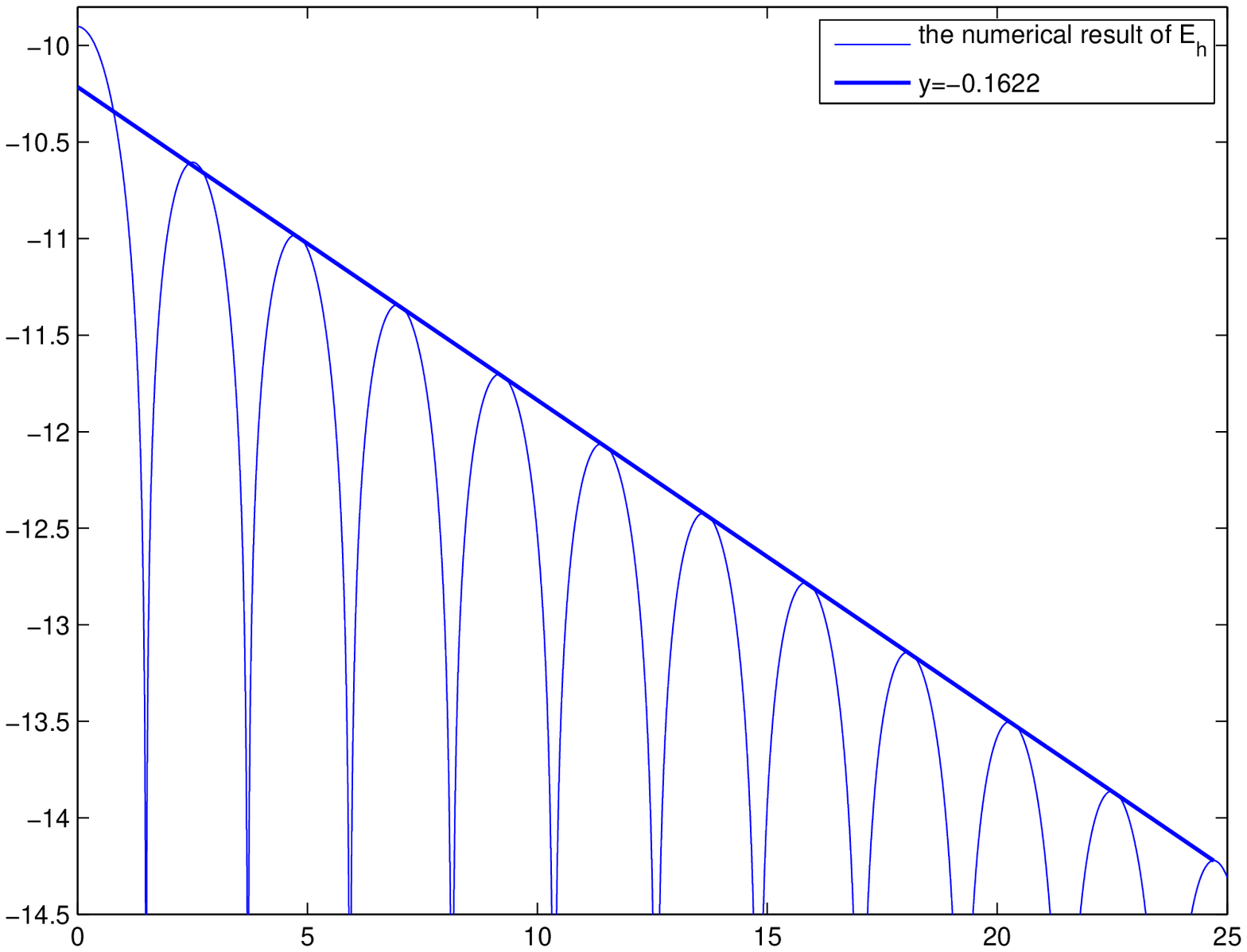}
\end{overpic} }
\caption{The dependence of the damping rate on the wave number $k$ in the
  collisionless case. The curves are the evolving of the square root
  of $\mathcal{E}_h$ in time using logarithm scale. The slope of the
  line is obtained by the least square fitting.}
\label{fig:no_collision}
\end{figure}

\begin{figure}[!ht]
\centering
\subfigure[$\nu = 0.01,k=0.2$]{
\begin{overpic}[width=.45\textwidth,height=5cm]{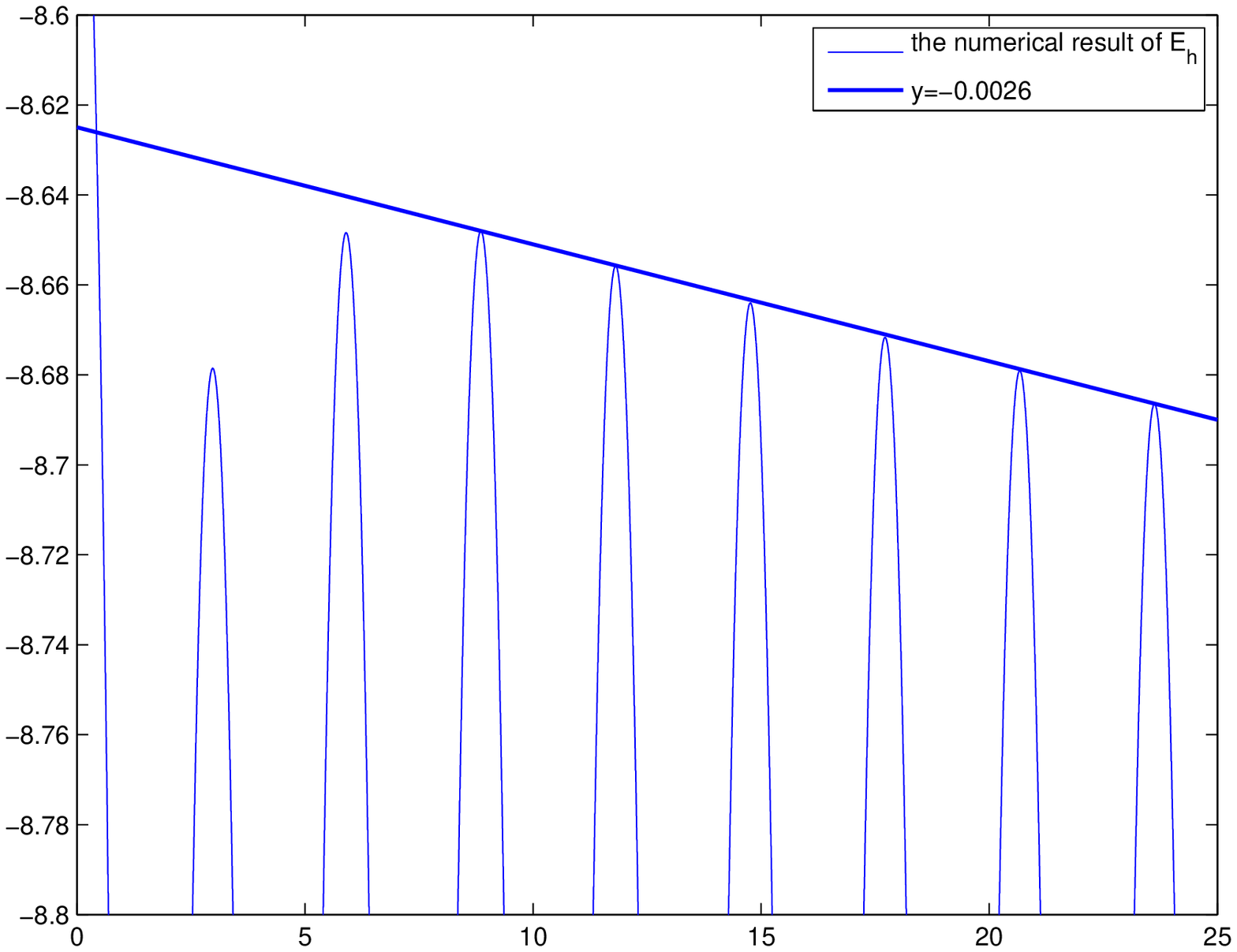}
\end{overpic}
}
\subfigure[$\nu =  0.01,k=0.3$]{
\begin{overpic}[width=.45\textwidth,height=5cm]{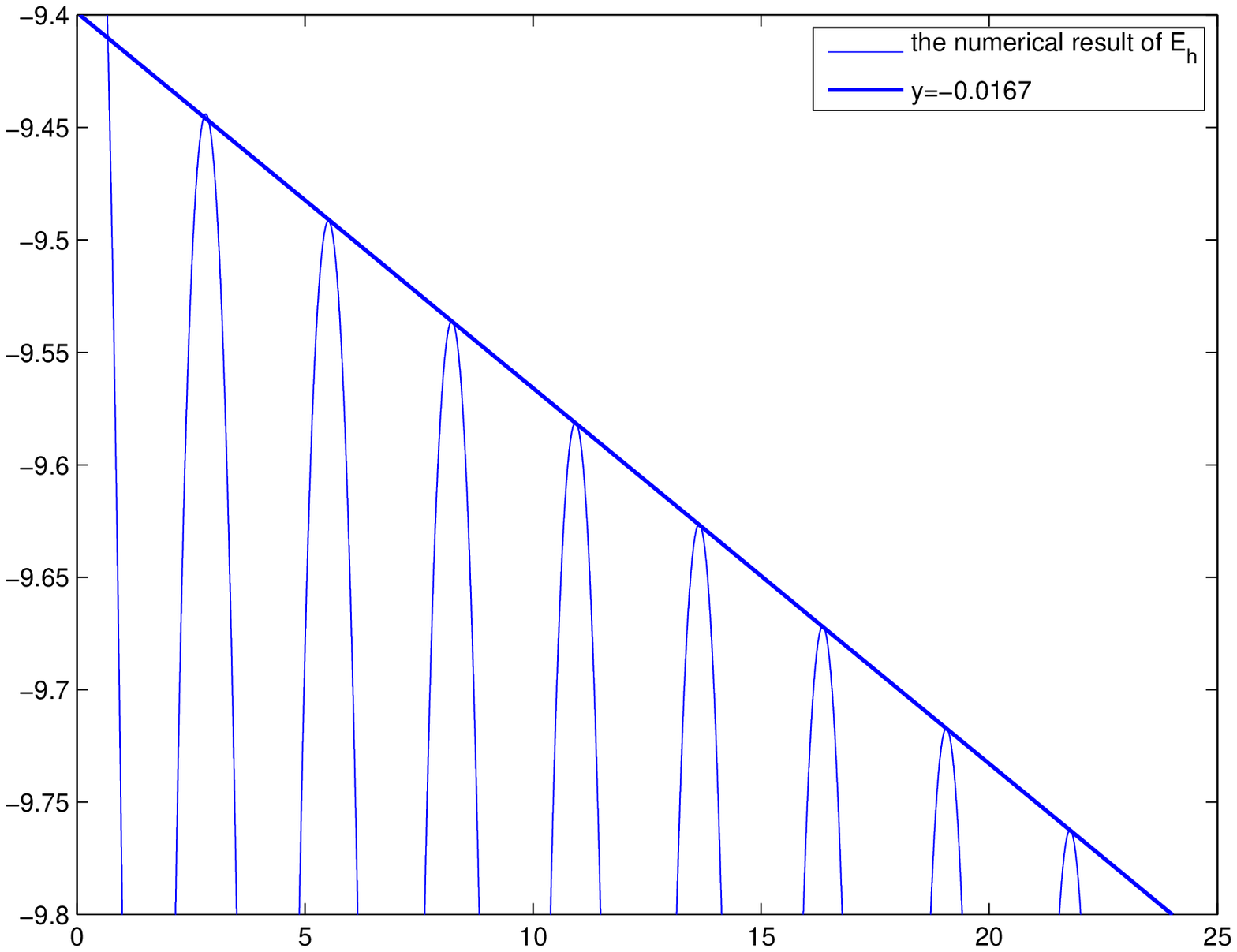}
\end{overpic}
}\\
\subfigure[$\nu = 0.01,k=0.4$]{
\begin{overpic}[width=.45\textwidth,height=5cm]{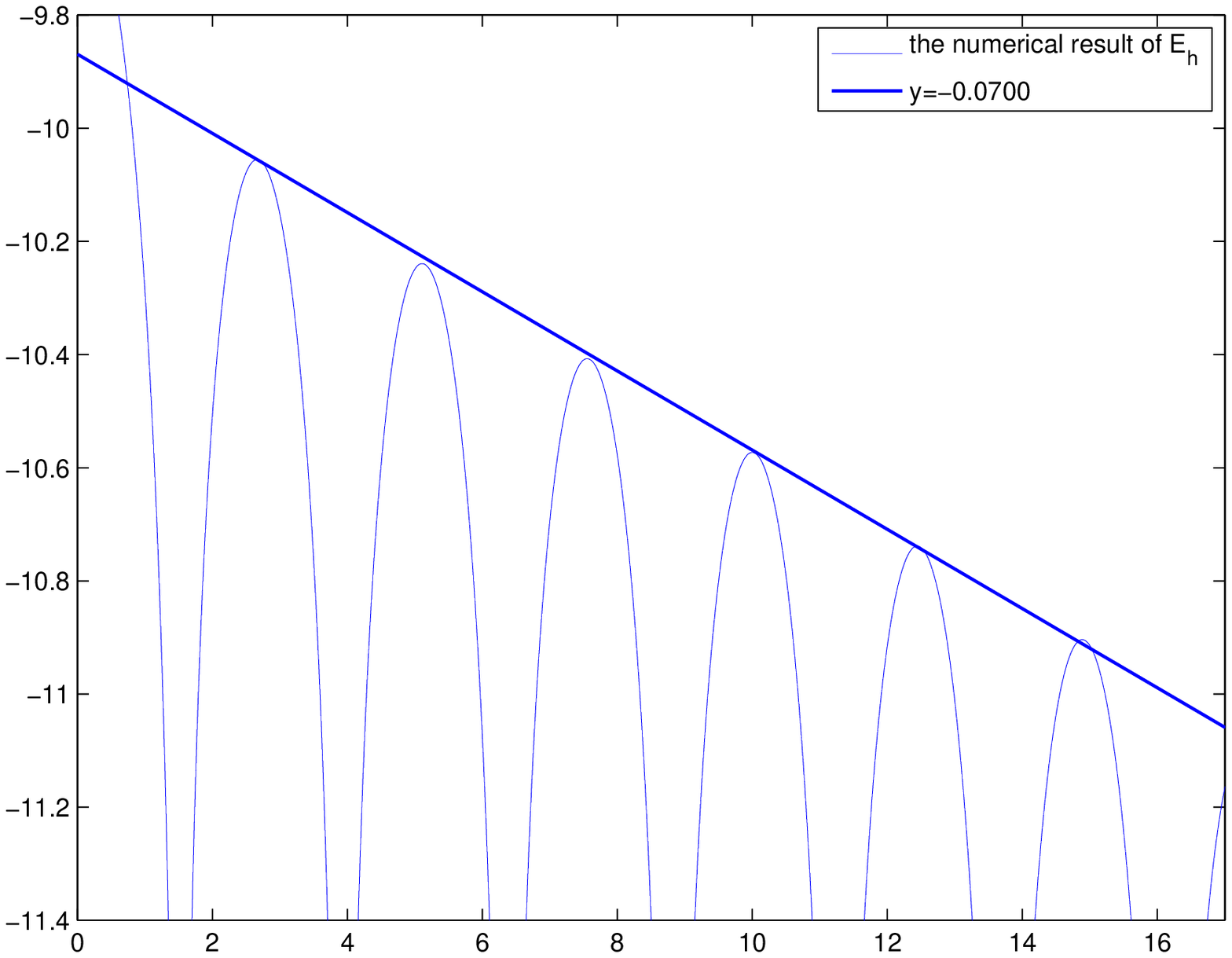}
\end{overpic}
}
\subfigure[$\nu = 0.01,k=0.5$]{
\begin{overpic}[width=.45\textwidth,height=5cm]{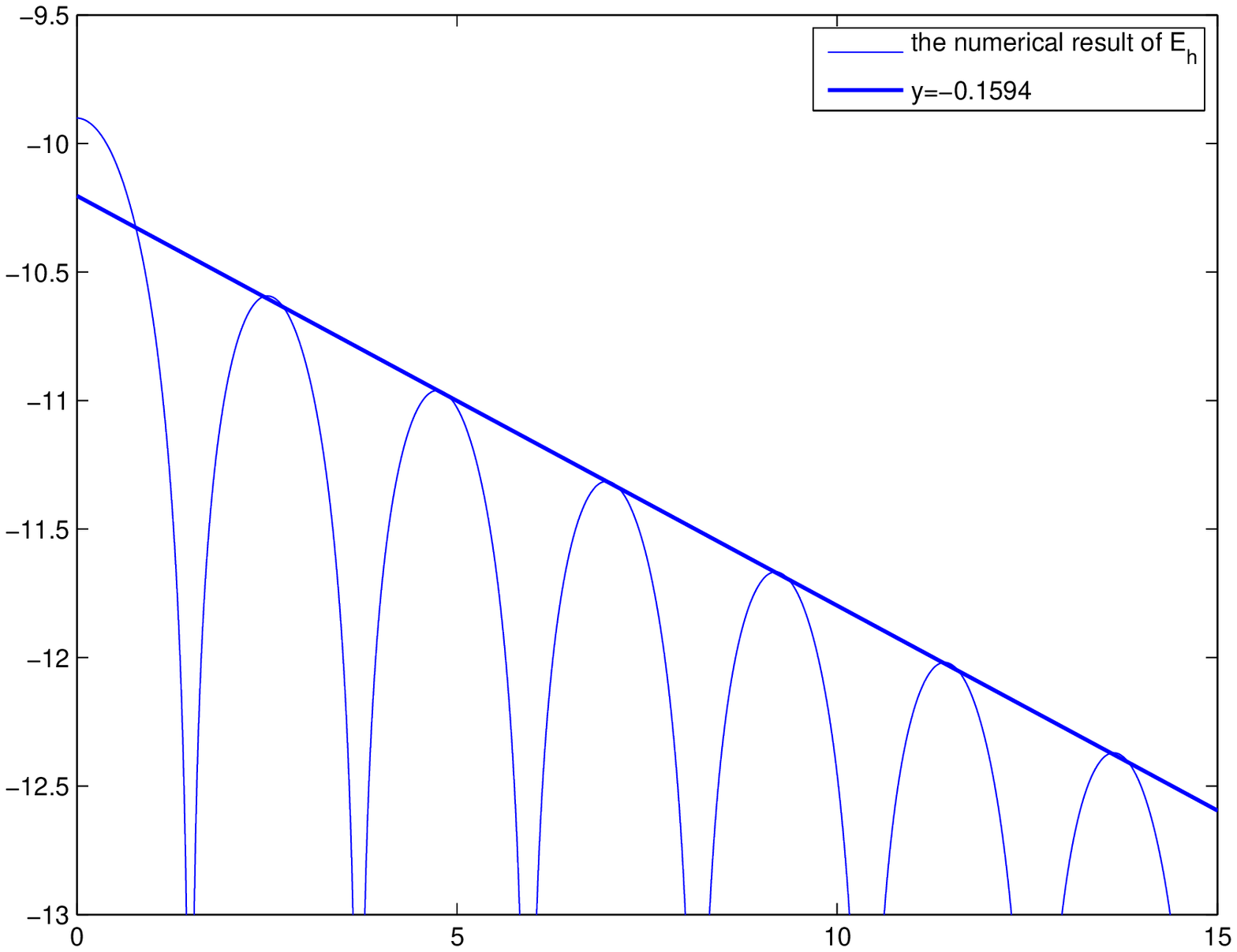}
\end{overpic}
}
\caption{The dependence of the damping rate on the wave number $k$ with
  collision frequency $\nu = 0.01$. The curves are the evolving of the
  square root of $\mathcal{E}_h$ in time using logarithm scale. The
  slope of the line is obtained by the least square fitting.}
\label{fig:collision}
\end{figure}

\begin{figure}[!ht] \centering
\subfigure[$\nu = 0.05,k=0.2$]{
\begin{overpic}[width=.45\textwidth,height=5cm]{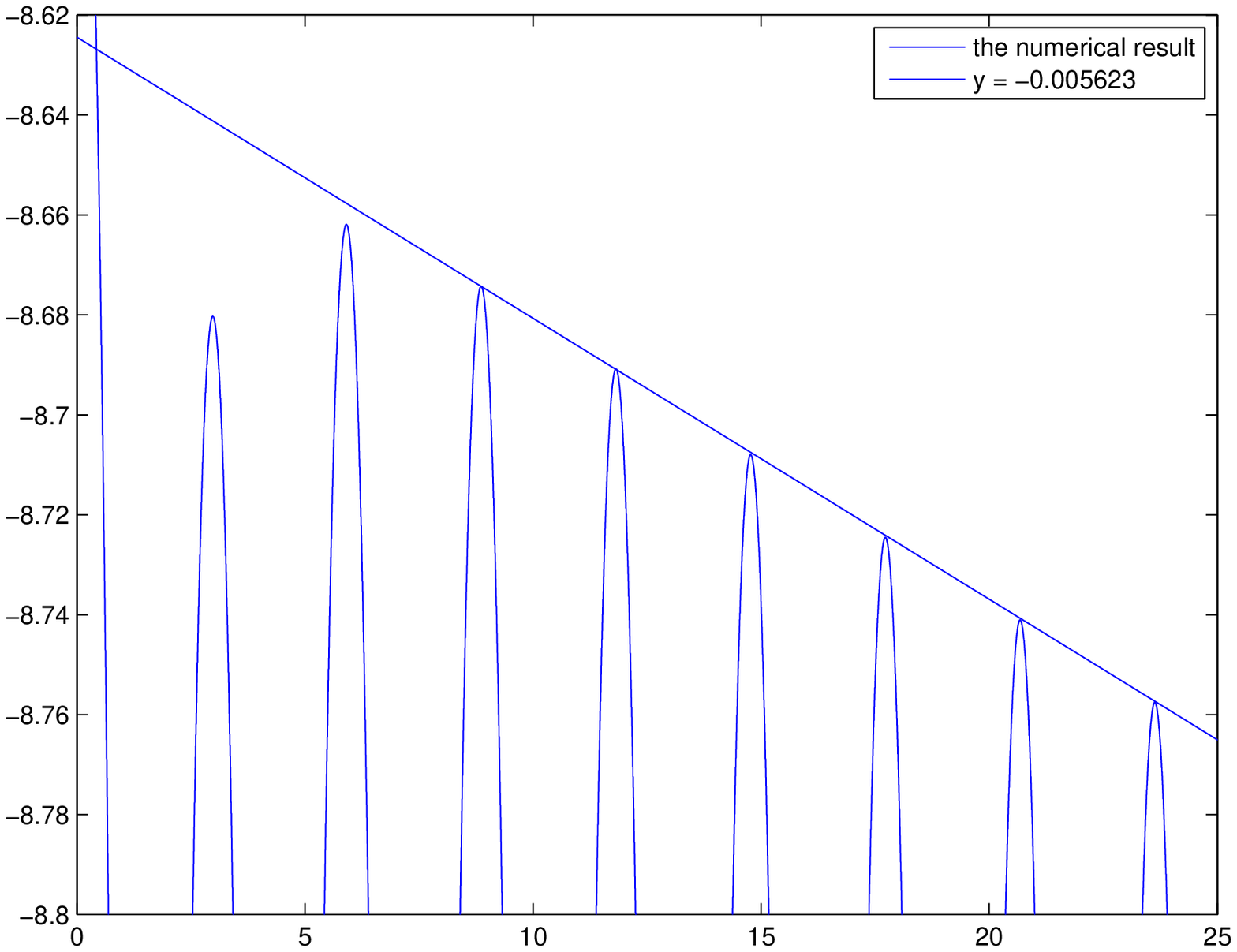}
\end{overpic}
}
\subfigure[$\nu = 0.05,k=0.3$]{
\begin{overpic}[width=.45\textwidth,height=5cm]{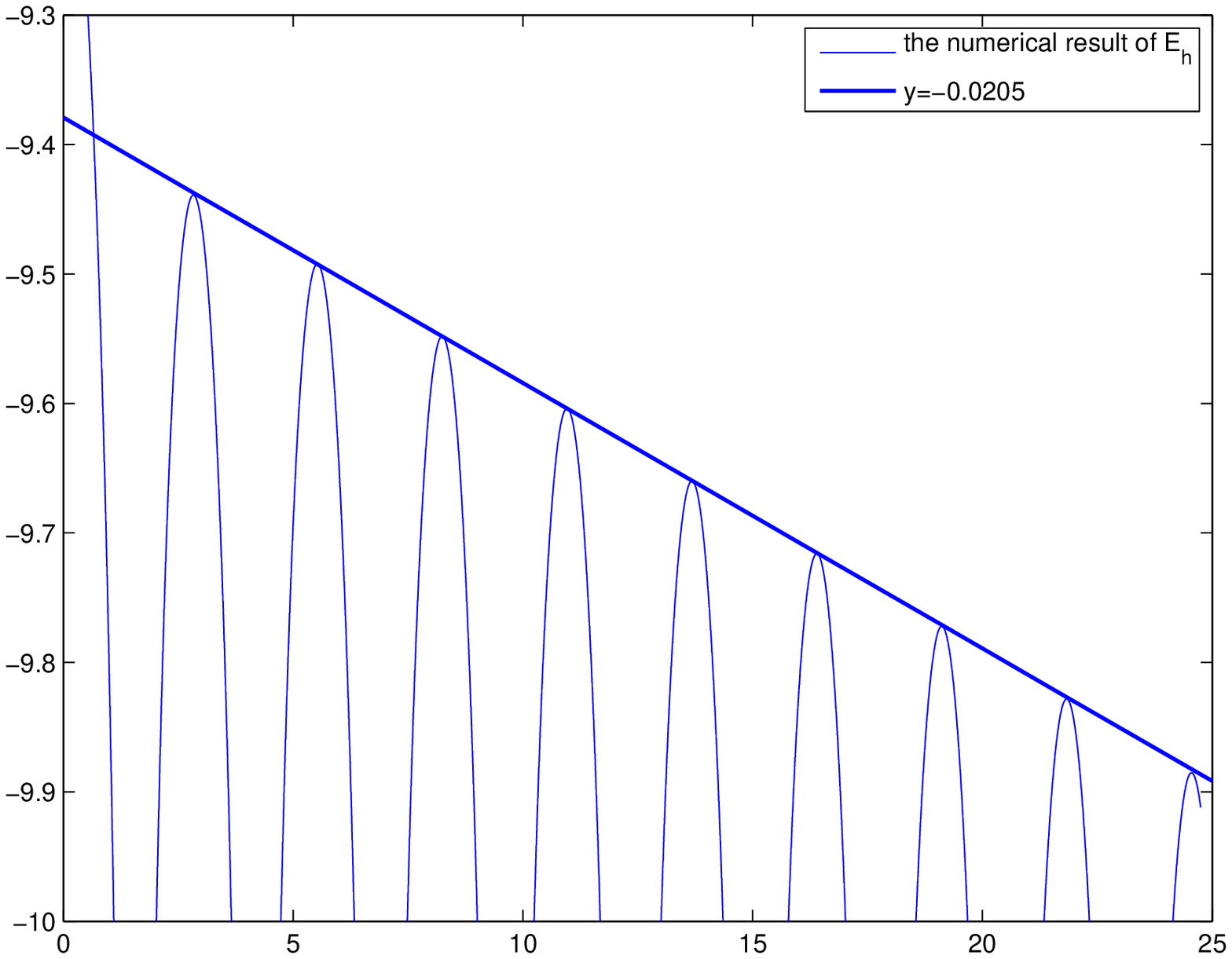}
\end{overpic}
} \\
\subfigure[$\nu = 0.05,k=0.4$]{
\begin{overpic}[width=.45\textwidth,height=5cm]{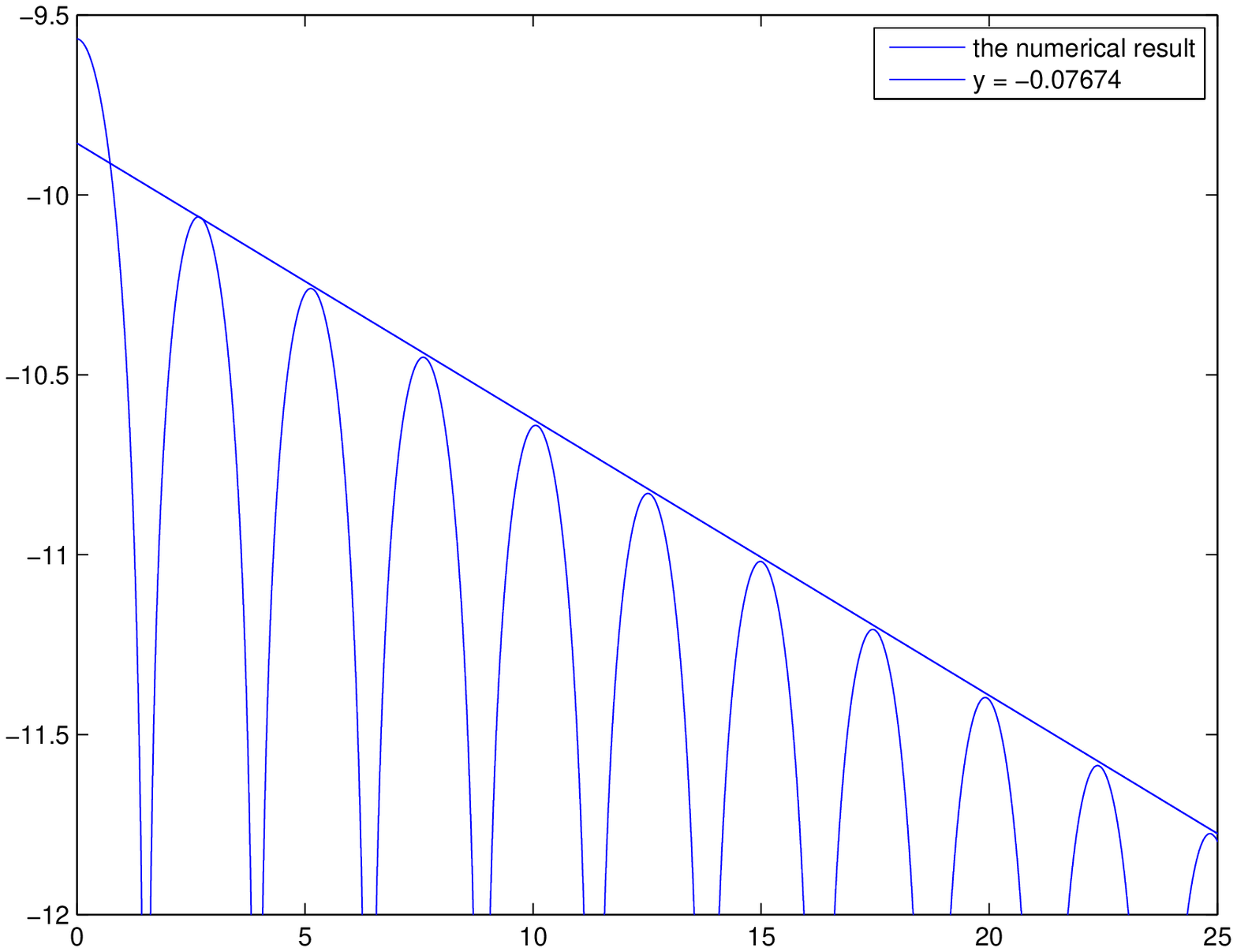}
\end{overpic}
}
\subfigure[$\nu = 0.05,k=0.5$]{
\begin{overpic}[width=.45\textwidth,height=5cm]{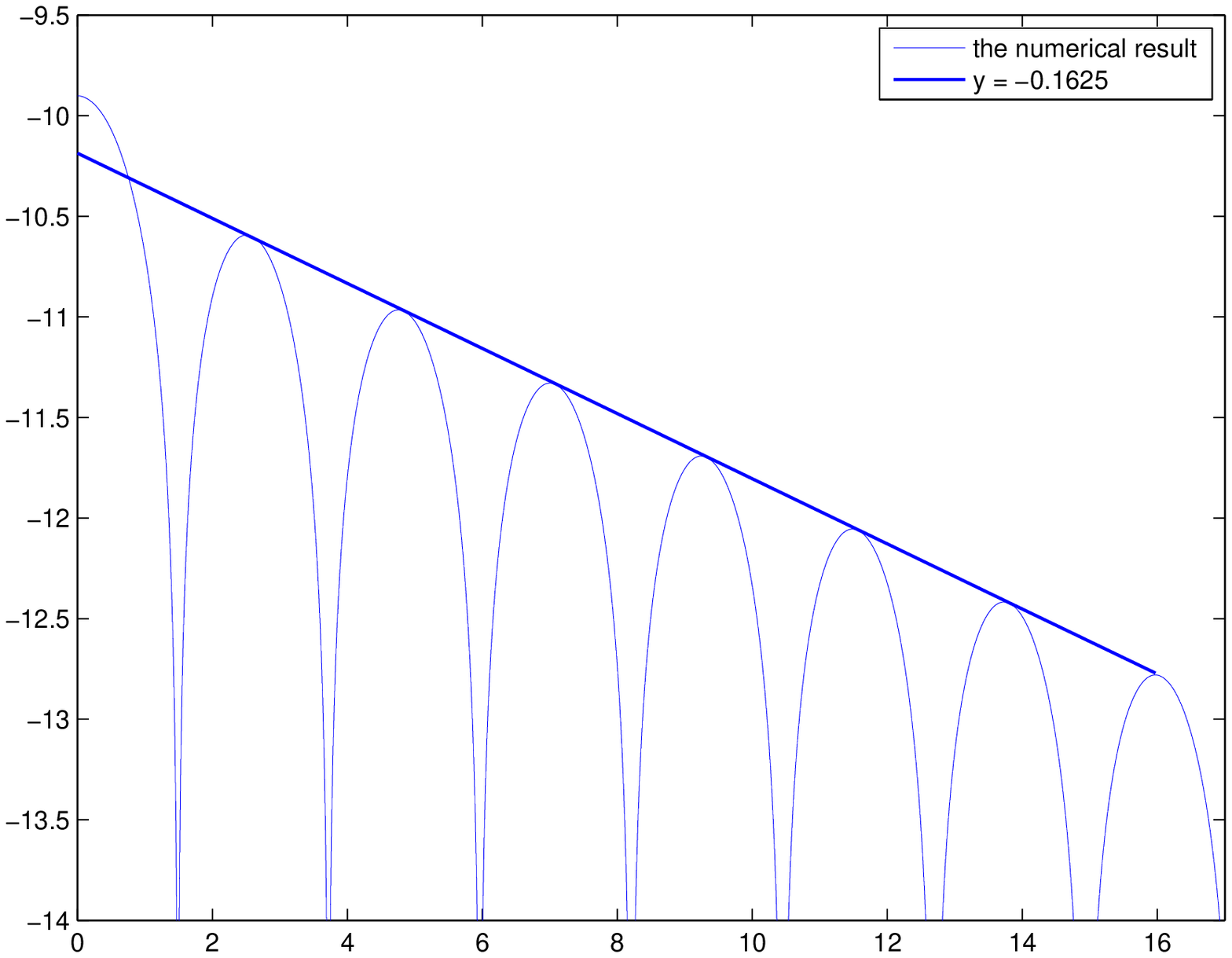}
\end{overpic}
}
\caption{The dependence of the damping rate on the wave number $k$ with
  collision frequency $\nu = 0.05$. The curves are the evolving of the
  square root of $\mathcal{E}_h$ in time using logarithm scale. The
  slope of the line is obtained by the least square fitting.}
\label{fig:nu=0.05}
\end{figure}

The Landau damping with different wave numbers and collision
frequencies are studied in this section. It is found that the
numerical phenomena here are exactly the same as the collisionless
case presented in Section \ref{sec:convergence}, for the wave numbers
ranging from $0.2$ to $0.5$ and the collision frequencies ranging from
$0.01$ to $0.05$. We have checked the numerical convergence both in
spatial grid size and order of moment expansion, and satisfied results
are obtained. In all cases, the exponential damping of the square root
of $\mathcal{E}_{h}(t)$ is observed.

Let us study the behavior of the solution with different wave numbers
firstly. We use $3000$ spatial grids and $80$ moments to compute the
evolution of $\mathcal{E}_h(t)$ for the wave numbers $k =0.2$, $0.3$,
$0.4$, and $0.5$ with the collision frequency $\nu = 0$, and the
results are in Figure \ref{fig:no_collision}. It is clear the
$\mathcal{E}_{h}(t)$ is damping exponentially in time and the damping
rate is increasing with greater wave number. This consists with the
empirical formula given in \cite{Filbet}.

In Figure \ref{fig:collision}, the evolution of the square root of
$\mathcal{E}_{h}$ with the wave number $k=0.2$, $0.3$, $0.4$ and $0.5$
with the collision frequency $\nu = 0.01$ is presented. We find that
$\mathcal{E}_{h}(t)$ is also damping exponentially but with greater
rates than that in the collisionless case. In Figure
\ref{fig:nu=0.05}, the numerical results for the wave number $k=0.2$,
$0.3$, $0.4$ and $0.5$ with a greater collision frequency $\nu = 0.05$
are plot. We find that $\mathcal{E}_{h}(t)$ is damping exponentially
with an even greater rate than that in the case of $\nu =0.01$. This
consists again with the empirical formula given in \cite{Filbet},
though a different collision term was used therein.

\begin{figure}[!ht]
\centering
\subfigure[$k=0.3,\nu=0.01$]{
\begin{overpic}[scale = .4]{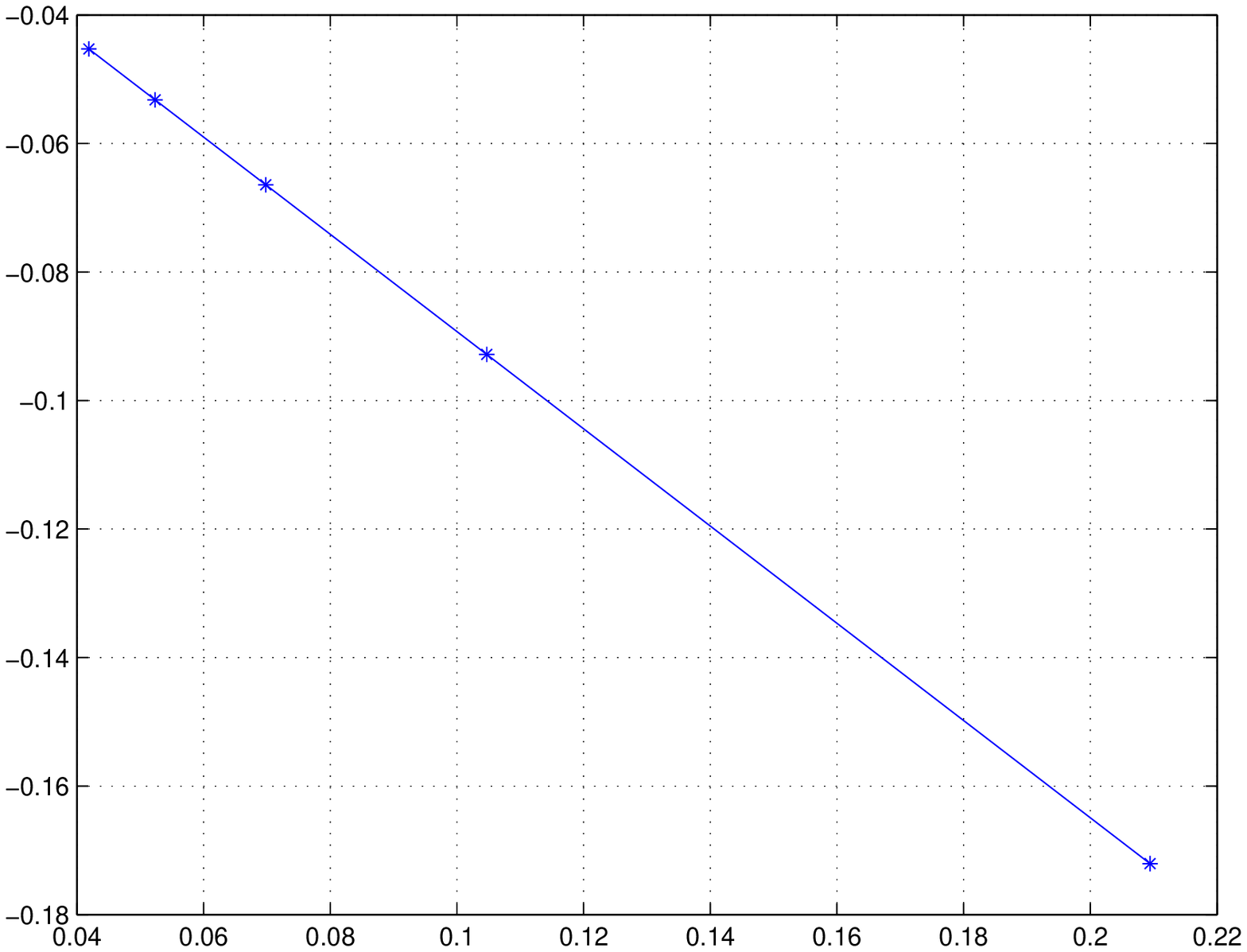}
\end{overpic}
}
\subfigure[$k=0.5,\nu=0.01$]{
\begin{overpic}[scale = .4]{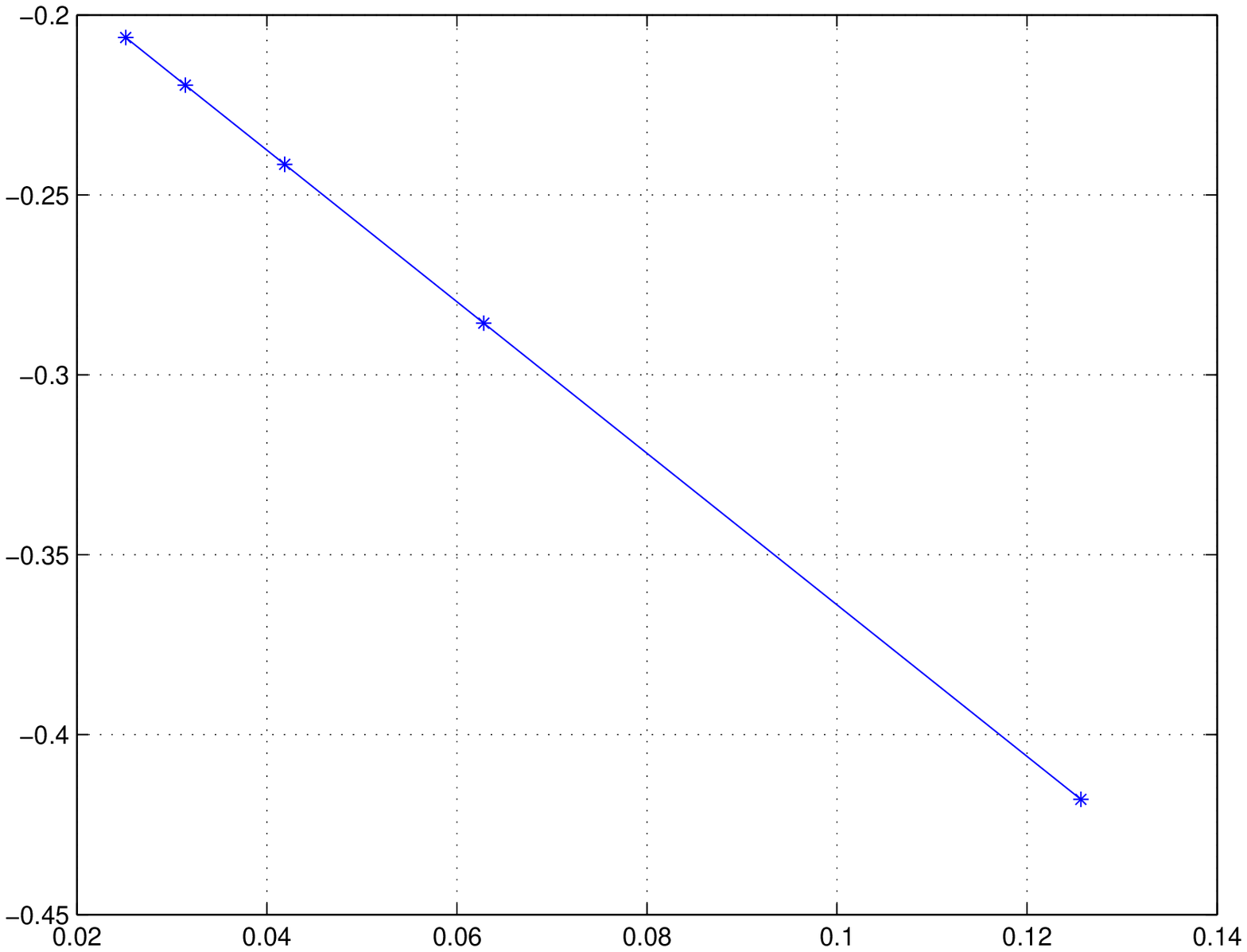}
\end{overpic} 
} \\
\subfigure[$k=0.3,\nu=0.05$]{
\begin{overpic}[scale = .4]{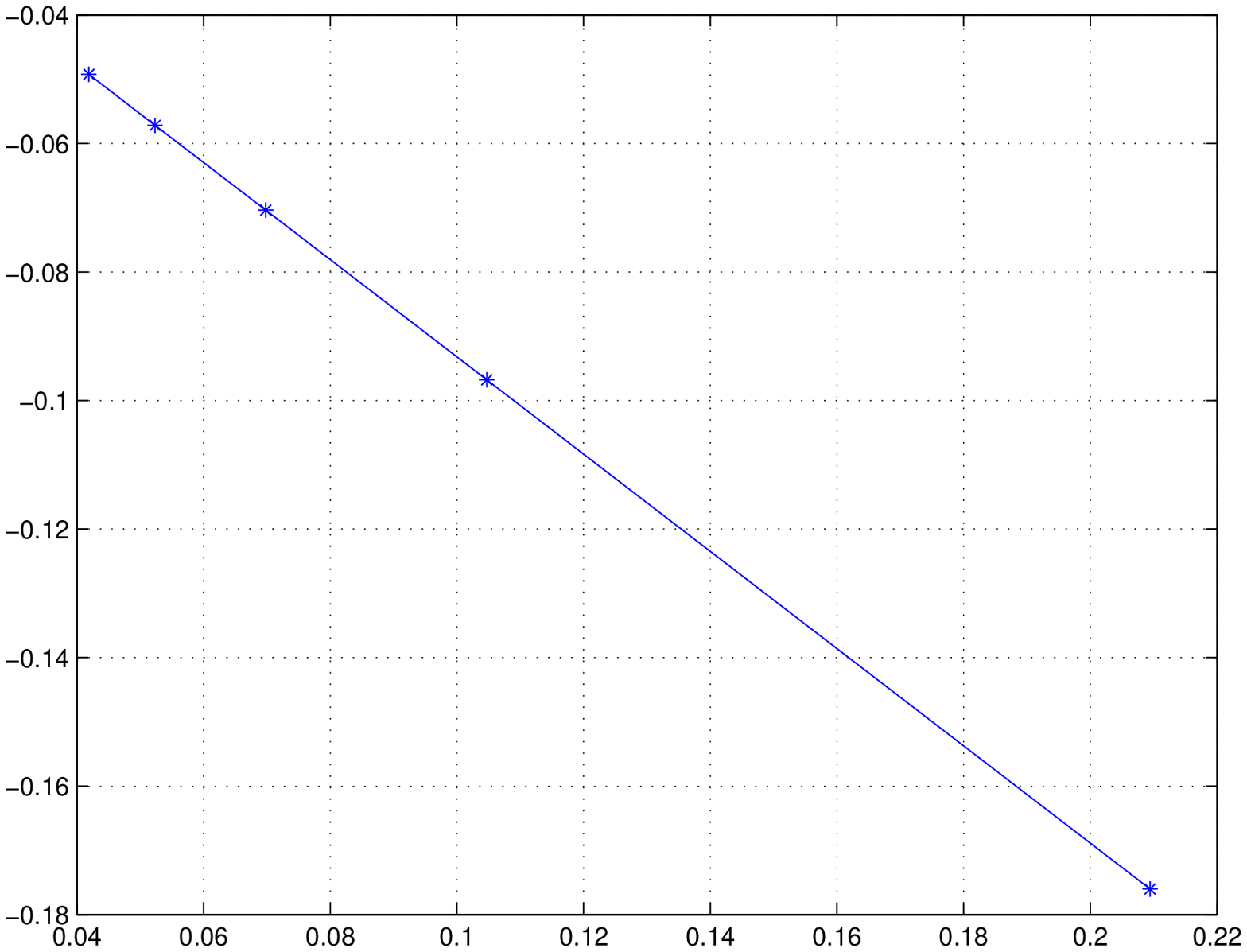}
\end{overpic}
}
\subfigure[$k=0.5,\nu=0.05$]{
\begin{overpic}[scale = .4]{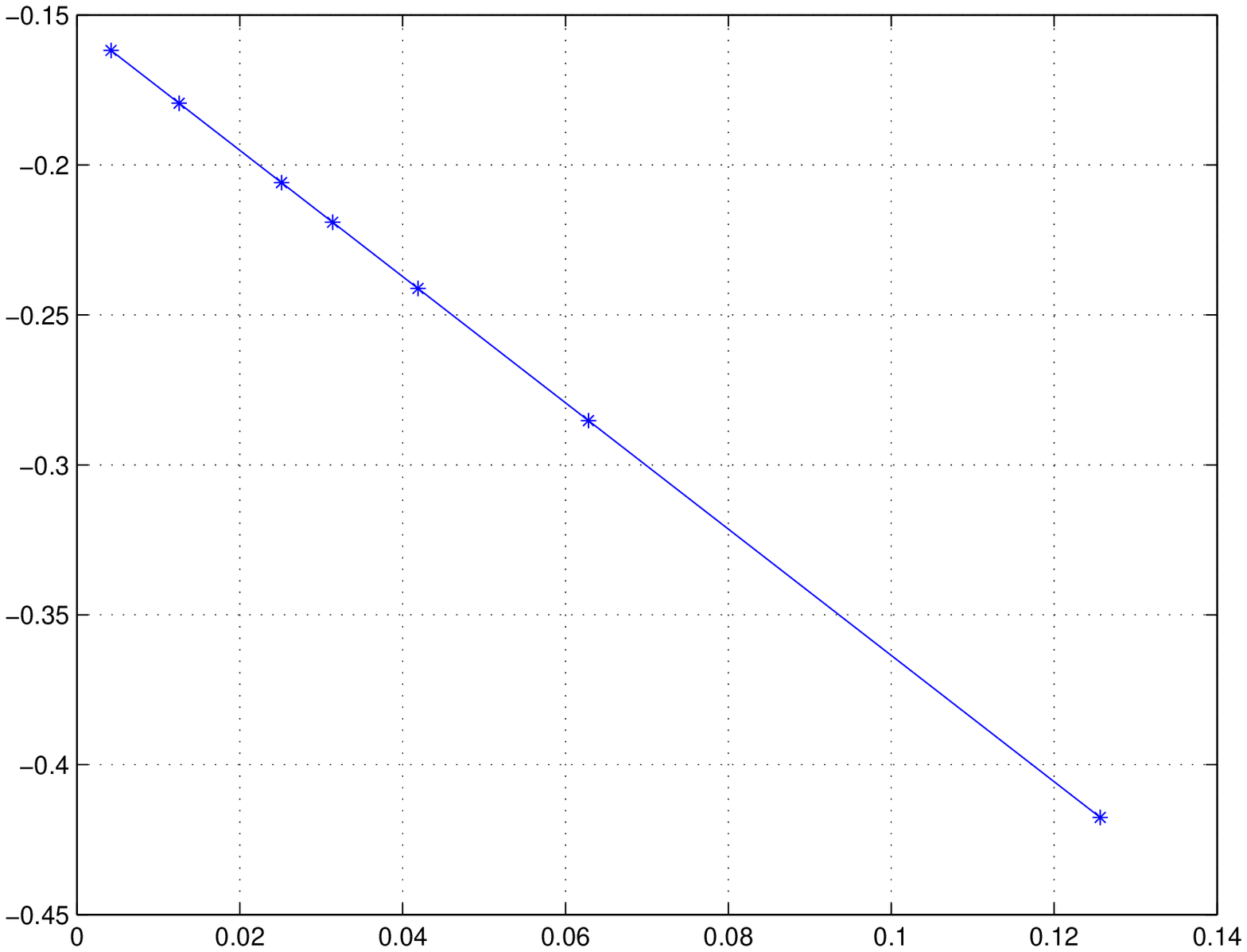}
\end{overpic}
}
\caption{The prediction of the exact damping rate using extrapolation
  \eqref{extrapolate_gamma}. The $x$-axis is the spatial grid size
  $\Delta x$ and the $y$-axis is numerical damping rates
  $\gamma_L^h(\Delta x)$.}
\label{fig:collision_slope}
\end{figure}

We numerically solve the V-B equations in 1D spatial space and 1D
velocity space with the number of moments fixed as $80$. The number of
the spatial girds is ranging form $100$ to $8000$. The least square
fitting of the peak value points is again adopted to get the numerical
damping rate of the square root of $\mathcal{E}_{h}$ for different
$\Delta x$. In Figure \ref{fig:collision_slope}, the damping rates for
the wave numbers $k=0.3$, $0.5$, and the collision frequencies $\nu =
0.01$, $0.05$ are plotted. It is clear that the numerical damping
rates are all in a linearly and monotonically converging pattern with
the spatial grid size $\Delta x$ going to zero. Remember that in the
collisionless case, the numerical damping rate is monotonically
linearly converging to the theoretic data with the spatial grid size
going to zero. Based on such similarity, it is reasonable to
conjecture that the damping rate with the collision term as the
function of the spatial grid size is monotonically linearly converging
to the exact damping rate while the spatial grid size is going to
zero, too. This inspires us to predict the exact damping rate for
different collision frequencies by the extrapolation of the numerical
damping rates, which is the same as the approach in the collisionless
case. We adopt the formula \eqref{extrapolate_gamma} again to retrieve
the parameters by least square fitting. The obtained parameter
$\gamma_L^{h,0}$ is regarded as the prediction of the exact damping
rate. In Table \ref{tab:limit_damping_rate}, we present
$\gamma_L^{h,0}$ for the wave numbers $k=0.2$, $0.3$, $0.4$, and $0.5$
with $\nu = 0.01$ and $0.05$.
\begin{table}[!ht]
  \centering
  \begin{tabular}[!ht]{|c|c|c|} \hline 
  Collision frequency $\nu$ & Wave number $k$ & $\gamma_L^{h,0}$ \\ \hline\hline
  \multirow{4}*{$0.01$} & $0.2$ &  $-8.3796\times10^{-5}$ \\ \cline{2-3} 
               & $0.3$  & $-0.01361$  \\ \cline{2-3}
               & $0.4$  & $-0.06701$  \\ \cline{2-3}
               & $0.5$  & $-0.15331$  \\ \hline\hline   
 \multirow{4}*{$0.05$} & $0.2$ &  $-0.002094$ \\ \cline{2-3} 
               & $0.3$  & $-0.01755$  \\ \cline{2-3}
               & $0.4$  & $-0.06966$  \\ \cline{2-3} 
               & $0.5$  & $-0.15230$  \\ \hline 
   \end{tabular}
   \caption{The prediction of the linear Landau damping rates.}
  \label{tab:limit_damping_rate}
\end{table}


%% file: article_conclusion.tex
\section{Concluding remarks}
The {\NRxx} method has been extended to solve the V-P and V-B
equations. The method is able to capture the linear Landau damping
effectively. We attempt to predict the exact Landau damping rates with
the collision term based on our observation in the collisionless case.

\section*{Acknowledgements}
This research was supported in part by the National Basic Research
Program of China (2011CB309704) and the Fok Ying Tong Education and
NCET in China.
